\documentclass{article}

\usepackage[english]{babel}
\usepackage[T1]{fontenc}    
\usepackage{amsmath}  
\usepackage{mathrsfs}
\usepackage{dsfont}
\usepackage{amssymb}     
\usepackage{bm}              
\usepackage{enumerate}    
\usepackage{array}            
\usepackage{xcolor}          
\usepackage[hyphens]{url} 
\usepackage{multicol}       
\usepackage{amsmath}
\usepackage{amsthm}
\usepackage[utf8]{inputenc}
\usepackage{tabularx}
\usepackage{longtable}
\usepackage{nicefrac}
\usepackage{microtype}

\newtheorem{theorem}{Theorem}[section]
\newtheorem{lemma}[theorem]{Lemma}

\theoremstyle{definition}
\newtheorem{definition}[theorem]{Definition}
\theoremstyle{remark}
\newtheorem{remark}[theorem]{Remark}

\numberwithin{equation}{section}

\newcommand{\abs}[1]{\left|#1 \right|}
\newcommand{\ainteg}[4]{\displaystyle{\int\limits^{#2}_{#1}}\! {#3}\,\text{d} {#4}}
\newcommand{\asum}[2]{\displaystyle{\sum\limits^{#2}_{#1}}}

\newcommand{\clim}[1]{\left[ {#1} \right]_\text{c}}
\newcommand{\dif}[1]{\frac{\text{d}}{\text{d} {#1}}}

\newcommand{\Exp}[1]{\text{e}^{#1}\,}
\newcommand{\fourier}[1]{\mathcal{F}[#1]}
\newcommand{\integ}[4]{\int\limits^{#2}_{#1}\!\! {#3}\,\text{d} {#4}}
\newcommand{\nab}[1]{\nabla_{\! #1}\,}
\newcommand{\norm}[1]{\lVert {#1} \rVert}

\newcommand{\state}[2]{\langle #1 \rangle_{#2}}

\DeclareMathOperator\co{co}

\DeclareMathOperator\si{si}

\DeclareMathOperator\supp{supp}

\title{\textbf{The Chaplygin Gas Equation of State for the Quantized Free Scalar Field on Cosmological Spacetimes}}
\author{\textsc{Jan Zschoche}  \\
	\\
	\emph{Institut f. Theoretische Physik, Universit\"at Leipzig}\\
	\emph{Postfach 100 920, D-04009 Leipzig, Germany}\\[0.25cm]
	\emph{Max-Planck-Institute f. Mathematics in the Sciences}\\
	\emph{Inselstr. 22, D-04103 Leipzig, Germany}\\
	\emph{jan.zschoche@gmx.de}
	}

\date{\today}

\begin{document}

\maketitle

\begin{abstract}
In this paper we try to answer the question whether the quantized free scalar field on a spatially flat Friedmann-Robertson-Walker (FRW) spacetime is a matter model that can induce a Chaplygin gas equation of state. For this purpose we first describe how one can obtain every possible homogeneous and isotropic Hadamard (HIH) state once any such state is given. We also identify a condition on the scale factor sufficient to entail the existence of a simple HIH state $-$ this state is constructed explicitly and can thence be used as a starting point for constructing all HIH states. Furthermore, we employ these results to show that on an FRW spacetime with non-positive constant scalar curvature there is, with one exception, no Chaplygin gas equation of state compatible with any HIH state. Finally, we argue that the semi-classical Einstein equation and the Chaplygin gas equation of state can presumably not be consistently solved for the quantized free scalar field. 
\end{abstract}


\section{Introduction}
One of the most challenging and surprising discoveries in cosmology within the last 15 years was the observation of an accelerated expansion of the universe. As gravity is usually acting as an attractive force one would have expected at least a decelerated expansion. In the standard model of cosmology, the $\Lambda$CDM model, this accelerated expansion is taken into account by the introduction of a non-vanishing cosmological constant $\Lambda$, where observations suggest a value of $\Lambda=10^{-52}\,\text{m}^{-2}$, cf. \cite{datagroup}. This cosmological constant effectively acts as a non-clustering energy form, usually referred to as dark energy.\\
One intensively investigated approach to explain the observed acceleration in a different way is the assumption of a certain equation of state (EoS) for the matter content of the universe, namely the \emph{Chaplygin Gas} (CG) EoS
\begin{equation}
\label{eqn_CG_EoS}
p = -\frac{A}{\varrho}, \quad A>0,
\end{equation}
we will now briefly describe. This EoS originates from the study of lifting forces on an airplane's wing and was first discovered by S. Chaplygin \cite{chaplygin} in 1904 and was later  on independently rediscovered in the same context by Tsien \cite{tsien} and von Karman \cite{vkarman}. That this EoS also yields very interesting results if applied to cosmological scenarios was first pointed out by Kamenshchik, Moschella and Pasquier in \cite{kamopa_1}, see also \cite{gor_ka_mo_1, gor_ka_mo_3, gor_ka_mo_2}. Let us assume the spacetime $(M,g)$ is of (a spatially flat) Friedmann-Robertson-Walker (FRW) type, i.e. $M= \mathbb{R}\times \mathbb{R}^3$ with line element $\text{d}s^2 = \text{d}t^2-a(t)^2(\text{d}x^2+\text{d}y^2+\text{d}z^2)$, where $t$ denotes the cosmological time and $a\in C^\infty(\mathbb{R},\mathbb{R}^+_0)$ the scale factor. For this spacetime the Einstein equation reduces to the Friedmann equations
\begin{equation}
\label{eqn_friedmann_equations}
\left(\dif{t}\, a(t)\right)^2 = \frac{8\pi}{3}\,\varrho(t)\,a(t)^2 \quad \text{and} \quad \dif{t} \left(\varrho(t)\,a(t)^3\right) = -p(t)\,\dif{t} a(t)^3,
\end{equation}
which can easily be solved using the CG EoS (\ref{eqn_CG_EoS}). As shown in \cite{kamopa_1} the solution reads in natural units as
\begin{align}
\label{eqn_cg_scalefactor}
\varrho(t) & =  \sqrt{A+\dfrac{B}{a_\text{cg}(t)^6}}\nonumber\\
 t   & =  \dfrac{1}{4\sqrt{6\pi}\,\sqrt[4]{A}}\left(\log\dfrac{\sqrt[4]{A+\frac{B}{a_\text{cg}(t)^6}}+\sqrt[4]{A}}{\sqrt[4]{A+\frac{B}{a_\text{cg}(t)^6}}-\sqrt[4]{A}}-2\arctan \sqrt[4]{1+\dfrac{B}{A\,\!a_\text{cg}(t)^6}} + b\right)\!,
\end{align}
where $b,B\in \mathbb{R}$ are constants of integration, and permits the following conclusions. First, the scale factor is zero for $t_*=(b-\pi)/(4 \sqrt{6\pi}\,\sqrt[4]{A})$ and thus $t_*$ determines the big bang singularity. For small scale factors, i.e. $a_\text{cg} \ll B/A$, one can expand the energy density with respect to small $a_\text{cg}$ and obtains for $\varrho$ the scaling behavior of a dust filled universe, namely $\varrho \approx \sqrt{B}\,a_\text{cg}^{-3}$ and $p\approx 0$. For large scale factors $a_\text{cg}\gg B/A$ on the other hand the energy density becomes basically a constant with $\varrho = -p = \sqrt{A}$ and thus mimics the effect of a cosmological constant. That is, with this simple EoS one can reproduce three important features of the $\Lambda$CDM model $-$ the big bang, the matter dominated era and the cosmological constant era $-$ in a generic way.\\
Note that it is even possible to incorporate a radiation dominated era by modifying the CG EoS via the introduction of additional parameters. For instance one could consider the so-called modified CG EoS $p= D\,\varrho-A\,\varrho^{-\alpha}$ with $D \geq 0$ and $1\geq \alpha >0$, cf. \cite{benaoum, debach}. While the $\alpha$ factor only leads to minor changes in the late time behavior of the scale factor, the additional term $D\cdot\varrho$ yields for small scale factors the radiation dominated era already announced, if $D=\frac{1}{3}$ is chosen.\\
However, this approach can only be considered successful if there are matter models actually inducing the CG EoS (or its modifications) since otherwise it would be merely a phenomenological description. So far the CG EoS arose in the context of string theory, where it can be obtained from the Nambu-Goto action for $d$-branes moving in a $(d+2)$-dimensional spacetime, cf. \cite{bord_hoppe, jackiw, kamopa_1}, and in classical field theory on curved spacetimes where the Lagrangian $L = \frac{1}{2}\dot{\phi}^2-\frac{\sqrt{A}}{2}(\cosh 3\phi + \cosh^{-1}3\phi)$ leads to the CG EoS, cf. \cite{kamopa_1}.\\[0.2cm]
In this paper we want to answer the question if the quantized minimally coupled free scalar field on a curved background is another matter model that leads to a CG EoS $-$ from a physical point of view the modified CG EoS mentioned above would be much more interesting but at the same would make any analysis much more complicated.\\
Therefore, we will give in the first part of the second section a brief introduction to quantum field theory (QFT) on curved spacetimes (CST). This in particular includes the description of the quantization scheme via the free field algebra and Hadamard states\footnote{Within the framework of QFT on CST these are the physically sensible states.} as well as the definition of the quantized energy-momentum tensor (EMT) $-$ this part contains no new results and is solely presented in order to fix notations and to make this work as self-contained as possible. In the second part of the second section, we will specify the introduced notions to the case of a spatially flat FRW background. That is, based on former results by L\"uders \& Roberts \cite{lued_rob} and Schlemmer \cite{schlemmer} on the mode decomposition of homogeneous and isotropic states as well as their renormalization, as described by Eltzner \& Gottschalk \cite{eltz_gott}, we explicitly calculate the renormalized quantized EMT for homogeneous and isotropic Hadamard (HIH) states in terms of a plane wave mode decomposition $-$ note that a similar mode decomposition for states of spin-0 and spin-1 fields on homogeneous spacetimes was recently achieved by Avetisyan \cite{avet} and Avetisyan \& Verch \cite{avet_verch}.\\
In the third chapter we study the occurrence of the CG EoS on different fixed background spacetimes.\\
We begin with Minkowski spacetime as the simplest case and very briefly argue that a CG EoS within the class of HIH states is only possible if two conditions are fulfilled. Namely, the energy density and pressure have to be constant and the remaining renormalization ambiguity has to be chosen properly.\\
In subsection \ref{sec_Homogeneous and Isotropic Hadamard States} we will include gravity by considering spacetimes with constant, non-positive scalar curvature $R$. Since not all homogeneous and isotropic states are physically sensible, i.e. they are not all Hadamard states, we have to describe in a first step how to obtain any HIH state once such a state is given $-$ this will be shown for a generic FRW spacetime. An initial ``Hadamard seed state'' for the latter procedure is constructed afterwards under a certain assumption on the scale factor in a rather simple manner in subsection \ref{sec_Construction_of_a_Hadamard_State_on_certain_spacetimes}. As this assumption is fulfilled by the spacetimes with scalar curvature $R=\text{const.}\leq 0$, as long as the minimally coupled scalar field is of mass $m=\sqrt{-\frac{1}{6}R}$, we are in the position to explicitly compute the EMT for all HIH states defined on these spacetimes. It will turn out that, unless $a(t)=a_\text{dS}(t) := \beta\, \Exp{\frac{m}{\sqrt{2}}t}, \beta \in \mathbb{R}^+,$ there are no HIH states entailing the CG EoS at all times, in particular late times. In the de Sitter case $a_\text{dS}$ there exists exactly one state compatible with the CG EoS, namely the Bunch-Davies state for which the energy density and pressure however become constants. Note that the analysis carried out explicitly for $R=\text{const.} \leq0$ and $m=\sqrt{-\frac{1}{6}R}$ is possible, since the field is effectively conformally coupled. It is therefore also possible to repeat the same analysis for the conformally coupled massless scalar field but with a somewhat \emph{different} EMT.\\
Based on the previous results, we will finally present some (yet heuristic) arguments which probably also exclude the existence of CG EoS inducing HIH states on an FRW spacetime with the CG scale factor $a_\text{cg}$. If made rigorous this would be equivalent to the claim that the semi-classical Friedmann equations and the CG EoS cannot be solved consistently within the class of HIH states and thus the quantized minimally coupled scalar field would not be a matter model compatible with the CG EoS.

\section{Quantum Field Theory in Curved Spacetime}
In this section we will give a brief introduction to QFT on CST. While in its first subsection the general framework is outlined we will specialize in the second subsection to FRW backgrounds. In particular, this includes the calculation of the quantized EMT in an HIH state.   

\subsection{The field algebra and Hadamard states}
\label{sec_The Field Algebra and Hadamard States}
As the general framework of quantum field theory on curved spacetimes is described in large detail elsewhere, cf. \cite{fewster_lec, hollands, wald_qft} and citations therein, we sketch only briefly the general setting to the extent we need, mainly for the sake of completeness and to fix our notation.\\[0.2cm]
In the following we will consider the minimally coupled free scalar field $\varphi$ on a globally hyperbolic spacetime $(M,g)$. The associated Klein-Gordon (KG) equation $P\varphi \equiv (g^{a b} \nab{a}\nab{b}+m^2)\,\varphi=0,\, m\geq0$, possesses, due to the global hyperbolicity of $(M,g)$, for every smooth and compactly supported initial data a unique solution. In particular, there are unique advanced and retarded Green's operators $E^\pm: C^\infty_0(M,\mathbb{R}) \rightarrow C^\infty(M,\mathbb{R})$, respectively, such that for all $f\in C^\infty_0(M,\mathbb{R})$ we have $P\,E^\pm f = E^\pm\, P f = f$ with $\supp(E^\pm f) \subset \mathcal{J}^\pm(\supp(f))$, $\mathcal{J}^\pm(O)$ denoting the causal future/past of the set $O$, cf. \cite{gin_baer_pfaef}. With these two operators we can define the fundamental solution $E:=E^+-E^-$ and associated to it the so-called \emph{commutator function}
\begin{equation*}
\mathscr{E}(f,h):=\integ{M}{}{f(q)\,(E\,h)(q)}{\mu_g(q)}\qquad \qquad \left(f,h \in C^\infty_0(M,\mathbb{R})\right),
\end{equation*}
where $\text{d}\mu_g(q):= \sqrt{\abs{g(q)}}\,\text{d}^4q$ denotes the metric-induced volume element on $M$. Since $E f$ solves by definition of $E$ the KG equation, $\mathscr{E}$ is a distributional bisolution of the field equation, i.e.
\begin{equation*}
\mathscr{E}(Pf,h) = \mathscr{E}(f,Ph) = 0.
\end{equation*}
Furthermore, due to $\supp(E f) \subset \mathcal{J}^+(\supp(f)) \cup \mathcal{J}^-(\supp(f))$ the commutator function vanishes, i.e. $\mathscr{E}(f,h) = 0,$ if the supports of $f$ and $h$ are acausally localized.\\
To quantize the system we consider the free unital *-algebra generated by the unit element $\mathds{1}$ and the field symbols $\phi(f)$ labeled by $f\in C^\infty_0(M,\mathbb{R})$. The \emph{algebra of free smeared fields} $\mathcal{A}(M,g)$ is then obtained by taking the quotient of the free unital *-algebra with respect to the relations\footnote{$[A,B]:=AB-BA$ denotes the commutator.} 
\begin{equation*}
\begin{array}{l r c l}
\textit{(}\mathit{1}\!\textit{)}\,  \text{\emph{Linearity}}:  		& \phi(\alpha f+h)  &\!\! = &\!\! \alpha \phi(f)+\phi(h)\\[0.2cm]
\textit{(}\mathit{2}\!\textit{)}\,  \text{\emph{Neutral field}}:		& \phi(f)^* 	    &\!\! = &\!\! \phi\left(\bar{f}\,\right)\\[0.2cm]
\textit{(}\mathit{3}\!\textit{)}\,  \text{\emph{Klein-Gordon Equation:}}	& \phi(Pf)  &\!\! = &\!\! 0\\[0.2cm]
\textit{(}\mathit{4}\!\textit{)}\,  \text{\emph{Commutation Relation:}}	& [\phi(f),\phi(h)] &\!\! = &\!\! -\text{\emph{i}}\,\mathscr{E}(f,h)\,\mathds{1},
\end{array}
\end{equation*}
where $\alpha \in \mathbb{C}$ and $f,h \in C^\infty_0(M,\mathbb{R}).$ Note that it is precisely the fourth relation that quantizes the field theory.\\
\\
As usual, a \emph{state} $\omega$ is given within this framework by a continuous linear functional $\omega: \mathcal{A}(M,g) \rightarrow \mathbb{C}$, $A \mapsto \omega(A) \equiv \state{A}{\omega}$, compatible with the *-structure of the algebra $\mathcal{A}(M,g)$, such that $\omega$ is normalized, i.e. $\omega(\mathds{1})=1$, and the positivity condition $\omega(A^*A)\geq 0$, for all $A\in \mathcal{A}(M,g)$, is fulfilled. Since this notion of a state is quite general, one cannot expect each state to be a physically sensible one. We therefore need to confine ourselves to a class of admissible states, which for instance enables us to define the energy-momentum tensor (EMT); for further discussions of this point cf. \cite{wald_qft}. Such a class is given by the set of (quasifree) Hadamard states we are going to characterize next.\\
The first restriction usually made is to consider only states whose $n$-point functions $$f_1\otimes ... \otimes f_n \mapsto \mathscr{W}^\omega_n(f_1,\;...\;,f_n):= \omega(\phi(f_1)...\phi(f_n))$$  are distributions. Such a state is completely characterized once all of its $n$-point functions are known. If additionally all odd $n$-point functions vanish while the even ones are given by 
\begin{equation*}
\mathscr{W}^\omega_{2n}(f,\,...\,,f) = \frac{(2n)!}{2^n\,n!}\, \mathscr{W}^\omega_2(f,f)^n,
\end{equation*} 
i.e. all $n$-point functions are completely fixed by the two-point function, we call the state \emph{quasifree}. It is furthermore a (quasifree) \emph{Hadamard state} iff for any convex normal neighborhood $\Omega\subset M$ of any point $q\in M$ and any time function\footnote{A smooth function $\tau:M\rightarrow \mathbb{R}$ is called time function if its gradient is past-directed timelike and its level sets are spacelike Cauchy surfaces.} $\tau,$ one can find a sequence $G^\omega_l\in C^l(\Omega\times\Omega,\mathbb{C})$ such that for all $f_1, f_2 \in C^\infty_0(\Omega,\mathbb{R})$ one has
\begin{equation}
\label{eqn_hadamard_prop_1}
\mathscr{W}^\omega_2(f_1,f_2) = \lim_{\varepsilon \rightarrow 0^+} \integ{\Omega\times\Omega}{}{\left(\mathcal{H}_{l,\varepsilon}(q,q')+G^\omega_l(q,q')\right)\,f_1(q)f_2(q')}{\mu_g(q)}\text{d}\mu_g(q'),
\end{equation} 
where 
\begin{align}
\mathcal{H}_{l,\varepsilon}(q,q') := & -\frac{u(q,q')}{4\pi^2\,(\sigma(q,q')+2\text{i}\,(\tau(q)-\tau(q'))\,\varepsilon+\varepsilon^2)}\nonumber\\[0.1cm]
\label{eqn_hadamard_prop_2}
				       & \quad-\frac{1}{4\pi^2}\,V_l(q,q')\, \log\left(-\frac{1}{\lambda^2}\,(\sigma(q,q')+2\text{i}\,(t(q)-t(q'))\,\varepsilon+\varepsilon^2)\right),\nonumber\\[0.1cm]
 V_l(q,q') := & \sum^l_{j=0} v_j(q,q')\,\sigma^j(q,q').
\end{align}
Here $\sigma(q,q')$ is the squared geodesic distance\footnote{If $\gamma:[a,b]\rightarrow \Omega$ denotes the unique geodesic connecting $q=\gamma(a)$ and $q'=\gamma(b)$ then $\sigma(q,q'):=\left(\integ{a}{b}{\sqrt{g(\dot{\gamma}(l),\dot{\gamma}(l))}}{l}\right)^2$. Note that $\sigma$ is negative, zero or positive whenever the geodesic is spacelike, lightlike or timelike, respectively.} between the two points $q$ and $q'$, $\lambda \in \mathbb{R}^+$ is an arbitrary constant length scale and the functions $u, v_j \in C^\infty(\Omega\times\Omega,\mathbb{R})$ are determined by the so-called Hadamard recursion relations\footnote{If not stated otherwise, Latin indices always label tensor components while Greek indices label tetrad components. Letters from the beginning of the alphabet, i.e. $a,b,c,...$ and $\alpha, \beta, \gamma, ...$, always run from $0,...,3$ while letters from the middle of the alphabet, i.e. $i,j,k,...$, only run from $1,...,3$.}  
\begin{align}
\label{eqn_had_par_u}
2\nab{a}\sigma \nabla^a\,u + (\square \sigma-8)\, u      & = 0\\[0.25cm]
\label{eqn_had_par_v0}
2\nab{a}\sigma \nabla^a\,v_0 + (\square \sigma-4)\, v_0 & = -(\square+m^2)\, u\\[0.25cm]
\label{eqn_had_par_vj}
2\nab{a}\sigma \nabla^a\,v_j + (\square \sigma+4(j-1))\, v_j & = -\frac{1}{j}\,(\square+m^2)\, v_{j-1} \quad \text{for}\quad 1\leq j\leq l
\end{align}
with the initial condition $u(q,q)=1$. For later purposes we define the symmetrization of the appearing distributions $\mathscr{W}^\omega_2(f_1,f_2)$, $\mathcal{H}_{k,\varepsilon}(f_1,f_2)$ and $G^{\omega}_l(f_1,f_2)$. Namely we set $\mathscr{W}^{\omega,\text{s}}_2(f_1,f_2):= \frac{1}{2}\,(\mathscr{W}^{\omega}_2(f_1,f_2)+\mathscr{W}^{\omega}_2(f_2,f_1)),$ and so on.\\
Note that, for arbitrary but fixed $l\in \mathbb{N}$, two different Hadamard states, say $\omega$ and $\omega'$, always share by definition the same singular part $\mathcal{H}_{l,\varepsilon}$. Hence, the difference of the two-point functions $\mathscr{W}^\omega_2(q,q')$ and $\mathscr{W}^{\omega'}_2(q,q')$ has to be $l$-times differentiable. But since $l$ was arbitrary we actually have $(\mathscr{W}^\omega_2(q,q')-\mathscr{W}^{\omega'}_2(q,q'))\in C^\infty(\Omega\times \Omega,\mathbb{C})$.\\
For a more in-depth discussion of Hadamard states we refer the reader for instance to \cite{kay_wald, radzikowski} for the quantized scalar field and for vector-valued quantum fields to \cite{sahlmann_verch2}.\\[0.2cm]
We can now turn to the construction of the quantized EMT where we follow the proposal of Moretti \cite{moretti}. That is, we first employ the point-splitting procedure, then, in order to preserve general covariance, subtract from the symmetrized two-point function $\mathscr{W}^{\omega,\text{s}}_2$ of a Hadamard state $\omega$ the symmetric Hadamard parametrix $H^\text{s}_{l,\varepsilon}$ of sufficiently high order thus obtaining a $C^l$ function, cf. \cite{wald_qft}. After applying a modified energy-momentum operator to this $C^l$ function we finally undo the point-splitting by taking the coincidence limit.\\
So let us define the canonical energy-momentum operator acting on a biscalar $\psi(q,q')\in C^2(M\times M)$ by
\begin{equation*}
T^\text{can}_{a b'} \psi(q,q'):= \left(\nab{a}\!\otimes \nab{b'}\! - \frac{1}{2}\,g_{a b'}\left(g^{c d'}\nab{c}\!\otimes\nab{d'}\! - m^2\,\mathds{1}\right)\right) \psi(q,q'),
\end{equation*}
where $g_{a b'}= g_{a c}\,g^{c}_{b'}$ and $g^{c}_{b'}$ denotes the parallel propagator. Then the quantized EMT is defined as follows.
\begin{definition}[Moretti \cite{moretti}]
\label{def_quantized_EMT}
Let $(M,g)$ be a globally hyperbolic spacetime, $\omega$ a quasifree Hadamard state and $P:=\square_q+m^2$ the Klein-Gordon operator acting with respect to $q$. Then the \emph{expectation value of the energy-momentum tensor} of the quantized minimally coupled Klein-Gordon field in the state $\omega$ is defined for $l\geq 1$ by
\begin{align}
\label{def_quantized_emt}
\state{T_{a b}(q)}{\omega} := & \clim{g_b^{\,b'}(q,q')\left(T^\text{can}_{a b'} -\frac{1}{3}\,g_{a b'}\, P \otimes \mathds{1}\right) \left(\mathscr{W}^{\omega,\emph{s}}_2(q,q')-\mathcal{H}^\emph{s}_{l,\varepsilon}(q,q')\right)}\nonumber\\[0.1cm]
					   &\,\, + C_{a b}(q),
\end{align}
where $\clim{...}$ denotes the coincidence limit $q'\rightarrow q$ and $C_{a b}(q)$ are the renormalization degrees of freedom defined by 
\begin{align}
\label{eqn_emt_ambiguities1}
C_{a b}(q) \equiv  D_1 & \,m^4\,g_{a b}(q) + D_2\,m^2\,G_{a b}(q)+ D_3\,I_{a b}(q) +D_4\,J_{a b}(q)\\
\text{with} \quad G_{a b}&  = R_{a b}-\frac{1}{2}\,R\,g_{a b}\nonumber\\
	     I_{a b}  & = g_{a b} \left(\dfrac{1}{2}\,R^2 + 2\,\square R\right) -2 \nab{b}\nab{a}R -2 R\, R_{a b}\nonumber\\
	     J_{a b}  & = \dfrac{1}{2}\,g_{a b}\left(R_{a b}\,R^{a b}+\square R\right) - \nab{b}\nab{a} R +\square R_{a b} -2 R_{c d}\,R^{c\;\;d}_{\;\;a\;\;b}\nonumber
\end{align}
and arbitrary constants $D_1,...,D_4 \in \mathbb{R}$, cf. \cite{hack, HoWa_2, wald_qft}.
\end{definition}
Note that, first, the above definition of $\state{T_{a b}(q)}{\omega}$ is actually independent of the specific choice of $l$ and, second, as was shown by Moretti in \cite{moretti}, this definition of the quantized EMT indeed obeys all four of Wald's axioms, cf. \cite{wald_2}. In particular, this means that the EMT is defined in a locally covariant way, it is covariantly conserved and a change of the length scale $\lambda$, entering via the Hadamard parametrix $\mathcal{H}^\text{s}_{l,\varepsilon}$, can be absorbed in the renormalization ambiguity $C_{a b}$ by a redefinition of the constants $D_1, ..., D_4$. Also there is precisely one length scale, say $\lambda_0$, such that $\state{T_{a b}}{\omega_0}$ vanishes if it is evaluated in the Minkowski vacuum state $\omega_0$.\\
Now we can use this EMT to formulate the semi-classical Einstein equation, which describes the back-reaction of the quantum field onto the spacetime.

\begin{definition}
A Hadamard state $\omega$ associated to a quantized scalar field and the metric of the underlying spacetime $(M,g)$ are said to fulfill the \emph{semi-classical Einstein equation} if
\begin{equation}
\label{eqn_sc_einstein_eqn}
R_{a b}-\frac{1}{2}\,R\,g_{a b} = 8\pi\, \state{T_{a b}}{\omega}.
\end{equation}
\end{definition}
Note that this equation only becomes meaningful as the definition of the EMT is locally covariant in the sense of \cite{brun_fred_ver, HoWa_1, HoWa_2, wald_qft}. Otherwise, it would have been meaningless to equate the EMT with the locally and covariantly defined Einstein tensor $G_{a b}= R_{a b}-\frac{1}{2}\,R\,g_{a b}$.  However, one should also notice that it is in general not possible to define a covariant map from the category of globally hyperbolic spacetimes into the class of Hadamard states.  This means that there exists no ``natural'' Hadamard state, which is defined on all spacetimes, cf. \cite{fewster_verch}. Therefore, strictly speaking, the EMT is in general not completely locally covariant and the semi-classical Einstein equation is in principle ill-defined $-$ on FRW spacetimes this problem can be fixed though, since on those spacetimes there exists at least one generic Hadamard state, namely the state of low energy, cf. \cite{degner, olbermann}. \\[0.2cm]
As we want to answer the question of the existence of the Chaplygin gas equation of state (CG EoS) in a cosmological scenario we will explicitly calculate in the next subsection the EMT for a homogeneous and isotropic Hadamard (HIH) state on an FRW spacetime.\\[-0.2cm]

\subsection{The energy-momentum tensor for homogeneous and isotropic Hadamard states on Friedmann-Robertson-Walker spacetimes}
\label{sec_EMT_HIH_state}
Motivated by cosmological observations we will exclusively work throughout the rest of the paper in a four-dimensional spatially flat FRW spacetime. That is, $(M= I\times\Sigma,g)$ is a four-dimensional, oriented and time-oriented Lorentzian manifold where $I$ denotes any interval and $\Sigma$ denotes the three-dimensional plane $-$ each point $q\in M$ is characterized by its coordinates $q=(t,\mathbf{x})$ with $\mathbf{x}\in \mathbb{R}^3$ and $t\in \mathbb{R}$  being the cosmological time. With respect to these coordinates the metric $g$ is given by
\begin{equation}
\label{eqn_metric_FRW}
g=\text{d}t\otimes\text{d}t-a(t)^2\,\delta_{ij}\,\text{d}\mathbf{x}^i\otimes\text{d}\mathbf{x}^j,
\end{equation}
where $\delta_{ij}$ denotes the Kronecker delta and $a \in C^\infty(\mathbb{R},\mathbb{R}^+_0)$ the scale factor. The associated volume element is given by $\text{d}\mu_g(q) = a(t)^3\,\text{d}t\,\text{d}^3\mathbf{x}$. An observer $\boldsymbol{\gamma}$, whose tangent vector field satisfies $\dot{\boldsymbol{\gamma}}(t) = \partial_t$, is called an \emph{isotropic observer} as it is the unique observer for whom the universe appears spatially isotropic. Associated to each isotropic observer is the orthonormal frame $\{e_0, e_1, e_2, e_3 \}:= \{\partial_t,a(t)^{-1}\,\partial_{1},a(t)^{-1}\,\partial_{2},a(t)^{-1}\,\partial_{3} \}$, where the frame components of the EMT $T_{a b}$ define, for instance, the \emph{energy density} $\varrho:= T_{a b}\,(e_0)^a\,(e_0)^b$ or the pressures $p_\mu:= T_{a b}\, (e_\mu)^a\, (e_\mu)^b$, $\mu=1,2,3,$  as measured by the isotropic observer.
Since the spatial section $\Sigma$ of a spatially flat FRW spacetime is just the ordinary $\mathbb{R}^3$, we can Fourier transform any test function $f=f(t, \mathbf{x})\in C^\infty_0(M,\mathbb{R})$ with respect to $\mathbf{x}$, i.e.
\begin{equation*}
\fourier{f}(t,\mathbf{k}) := \integ{\mathbb{R}^3}{}{f(t,\mathbf{x}) \,\Exp{\text{i}\, \mathbf{kx}}}{^3 \mathbf{x}}.
\end{equation*}
Note that $\fourier{f}$ depends only on the modulus $k$ of the vector $\mathbf{k} \in \mathbb{R}^3$, if this already holds for $f$ with respect to $\mathbf{x}$, i.e. if $f(q)=f(t,x)$ with $x=\norm{\mathbf{x}}$ being the Euclidean norm then $\fourier{f}=\fourier{f}(t,k)$.\\
To compute the quantized EMT, cf. equation (\ref{def_quantized_emt}), for an HIH state on an FRW spacetime we will proceed as follows. First, we will cite a theorem, based on the work of L{\"u}ders \& Roberts \cite{lued_rob} and Schlemmer \cite{schlemmer}, that characterizes all homogeneous and isotropic states by a plane wave mode decomposition. Then, based on the results of Eltzner \& Gottschalk \cite{eltz_gott}, we calculate the (equal time Fourier transformed) Hadamard parametrix for the energy density, as this is an essential part of the renormalization procedure. After determining the explicit form of the renormalization ambiguities $C_{a b}$ we finally put all results together in order to take the coincidence limit of $\left(T^\text{\emph{can}}_{0 0'} -\frac{1}{3}\,g_{0 0'}\, P \otimes \mathds{1}\right)(\mathscr{W}^{\omega,\emph{s}}_2-\mathcal{H}^\emph{s}_{1,\varepsilon})$ and thus obtain the quantized energy density. Afterwards we establish that the EMT in an HIH state is of perfect fluid type and use its covariant conservation to compute the pressure.
\begin{theorem}[L\"uders \& Roberts \cite{lued_rob}, Schlemmer \cite{schlemmer}]
\label{thm_homo_iso_states}
Suppose $\omega$ is a quasifree homogeneous and isotropic state of the minimally coupled free KG field in a spatially flat FRW background with scale factor $a(t)$ and Hubble parameter $H = \dot{a}/a$. Then $\omega$ is given by its two-point distribution
\begin{align}
\mathscr{W}^{\,\omega}_2(f,h) & \,= 
   \integ{\mathbb{R}^3}{}{\,\integ{\mathbb{R}}{}{\,\integ{\mathbb{R}}{}{\mathit{w}_k(t,t')\;\fourier{f}(t,\mathbf{k})\, \overline{\fourier{\overline{h}}}(t',\mathbf{k})\; a^3(t)\,a^3(t')}{t}}{t'}}{\!\; ^3\mathbf{k}}\label{eqn_2pd_his}\\[-0.2cm]
\text{with}\nonumber\\
\mathit{w}_k(t,t') &:=\frac{1}{(2\pi)^3}\,\left( \Xi(k)\,T_k(t)\,\bar{T}_k(t')+\left(\Xi(k)+1\right)\bar{T}_k(t)\, T_k(t')\right), 
\label{eqn_mode_hi_state}
\end{align}
where $T_k(t)$ fulfills the following differential equation and normalization condition
\begin{align}
\label{eqn_tkge}
\ddot{T}_k(t) +3H(t)\, \dot{T}_k(t)+ \omega_k(t)^2\,T_k(t) &= 0,\qquad\,\, \omega_k(t)^2 \equiv \frac{k^2}{a(t)^2}+m^2\\
\label{eqn_normkge}
\dot{T}_k(t)\, \bar{T}_k(t) - \dot{\bar{T}}_k(t)\, T_k(t) &= \text{\emph{i}}\,a(t)^{-3}.
\end{align}
Furthermore, $k\mapsto T_k(t)$ as well as $k\mapsto \dot{T}_k(t)$ are polynomially bounded. $\Xi(k)$ is an integrable polynomially bounded real-valued function, which is almost everywhere non-negative.\\[0.1cm]
If a quasifree state has a two-point function of the type just described then the state is homogeneous and isotropic.
\end{theorem}
\begin{remark}
\label{rem_bogoliubov_trafo}
As soon as one homogeneous and isotropic state $\omega_0$, characterized by the mode functions $T_k(t)$, is given one can construct any other homogeneous and isotropic state $\omega$, characterized by $S^\omega_k(t)$, out of $\omega_0$ via a \emph{Bogoliubov transformation}. That is, there are $k$-dependent constants $\alpha(k),\beta(k)\in \mathbb{C}$ such that $S^\omega_k(t) = \alpha(k)\,T_k(t)+\beta(k)\,\bar{T}_k(t)$ and $\abs{\alpha(k)}^2-\abs{\beta(k)}^2= 1$ hold.
\end{remark}
Now that we have specified the two-point distributions of homogeneous and isotropic states we need to turn to the Hadamard parametrix. In preparation of this section's main theorem about the quantized EMT in an HIH state we have to explicitly compute the Hadamard series (\ref{eqn_hadamard_prop_2}), apply to it the energy-momentum operator $T^\text{can}_{a b'}$ and finally take the equal time limit $t'\rightarrow t$ as well as its Fourier transform.\\
The main part of this task was already accomplished by Eltzner and Gottschalk in \cite{eltz_gott}. Namely, they first determined the squared geodesic distance and the Hadamard coefficients up to a sufficiently high order in $(t-t')$ and $(\mathbf{x-x'})$, calculated the necessary derivatives and expanded these for equal times in a power series with respect to the modulus of the coordinate difference $\mathbf{z}= (\mathbf{x}-\mathbf{x'}) \in \mathbb{R}^3$ $-$ due to homogeneity and isotropy of the FRW spacetime the squared geodesic distance $\sigma$ and the Hadamard coefficients $u$ and $v_j$ depend on the spatial coordinates solely via the Euclidean norm of the coordinate difference, $z:=\norm{\mathbf{x-x'}}$. Additionally, they discarded all terms in these expansions which vanish in the coincidence limit $z\rightarrow 0$, like $z^2$ or $z^2\,\log z$. Here we will indicate this operation of expanding a function $f(t,t',z)$ with respect to $z$, taking the equal time limit and discarding all ``regular'' orders in $z$, which are vanishing in the coincidence limit, by $[f(t,t',z)]^o_{t'=t}$.  Then they take the (distributional) Fourier transform of the leading order to obtain a form which is compatible with a mode decomposition of the two-point function. For the subleading singular orders they chose an ansatz for the structure of the Fourier transformed Hadamard series, motivated by their results of the leading order, and determined the free coefficient functions accordingly. \\
However, a drawback of this elegant approach is that one cannot fix the non-vanishing contributions from the $z^0$-order. We make up for this in the next lemma.
\begin{lemma}
\label{lem_h_rho/p}
Let $(M,g)$ be a spatially flat FRW spacetime with scale factor $a(t)$, Hubble parameter $H(t)=a(t)^{-1}\dot{a}(t)$, Ricci scalar $R(t)=-6(\dot{H}(t)+2H(t)^2)$ and $\clim{v_1}$ the coincidence limit of the Hadamard recursion coefficient $v_1(q,q')$. Then for a minimally coupled scalar field of mass $m\geq0$ and $l\geq 1$, the following expansions hold:
\begin{align}
\label{eqn_h_rho}
\mathfrak{h}^\varrho(t,k)&:=  \frac{1}{\pi^2}\,\clim{v_1}\,\delta_0(\mathbf{k})+\fourier{\left[(T^\text{can}_{00'}\mathcal{H}^\text{s}_{l,0})(t,t',z)\right]^o_{t'=t}}(t,k) \nonumber\\
					& = \asum{n=-1}{1}\mathfrak{h}^\varrho_{2n-1}(t)\, k^{2n-1} +\mathfrak{h}^\varrho_{0}(t)\,\delta_0(\mathbf{k})
\end{align}
where $\delta_0(\mathbf{k})$ is the three-dimensional delta distribution and the purely time dependent coefficients $\mathfrak{h}^{\varrho}_n(t),\, n\in\{1,0,-1,-3\}$  are given by 
\begin{align}
\mathfrak{h}^{\varrho}_1(t) 		& :=  \dfrac{1}{2 a^4}\nonumber\\
\mathfrak{h}^{\varrho}_{-1}(t)	& :=  \dfrac{m^2+H^2}{4\,a^2}\hspace{8.5cm} \nonumber\\
\mathfrak{h}^{\varrho}_{-3}(t)	& := \dfrac{1}{576 }\left(R^2-36m^4+12H^2(6m^2+R)+12H\dot{R}\right)\nonumber\\
\mathfrak{h}^{\varrho}_0(t)		& := -\,\dfrac{\pi}{4320}\left(36 H^4 - 11 R^2 + 12 H^2 (150 m^2 + 19 R) - 132 H \dot{R}\right)\nonumber\\
					 		&\qquad -\,\dfrac{\pi}{144}\,\log\left(\dfrac{a}{\lambda}\right)\left(R^2-36m^4+12H^2(6m^2+R)+12H\dot{R}\right).
\end{align}
The symbol $[...]^o_{t=t'}$ (basically) denotes the equal time limit, see the explanation above.
\end{lemma}
\begin{proof}
The proof is sketched in the appendix.
\end{proof}
In the next lemma we state the explicit form of the renormalization ambiguities for an FRW spacetime.
\begin{lemma}
\label{thm_renorm_ambiguity_rw_spacetime}
For a spatially flat FRW spacetime equipped with the metric (\ref{eqn_metric_FRW}) the renormalization ambiguity $C_{0 0}$, defined via equation (\ref{eqn_emt_ambiguities1}), is given by
\begin{align}
\label{eqn_renorm_ambiguity_rw_spacetime}
C_{00} 	& = D_1 m^4 + 24 D_4 H^4 + \frac{1}{6}\,(D_4-3 D_3 ) R^2 \nonumber\\
		&\hspace{0.5cm}	+ H^2 \left(3 D_2 m^2 + (4D_4-6 D_3) R\right) + 2 (3 D_3 + D_4) H\dot{R}.
\end{align}
Furthermore, we have $C_{11}=C_{22}=C_{33}$ and $C_{a b}=0$ for $a\neq b.$
\end{lemma}
\begin{proof}
Using definition (\ref{eqn_emt_ambiguities1}) for the renormalization ambiguities as well as the well-known Christoffel symbols for a spatially flat FRW spacetime this lemma can be proven by a straight forward calculation. 
\end{proof}
Finally, we can state this section's first central result, namely the complete energy density for HIH states.
\begin{theorem}
\label{thm_renorm_energy/pressure}
Let $(M,g)$ be the spatially flat FRW spacetime with scale factor $a(t)$ and $q\in M$. Furthermore, we choose an arbitrary $k_0>0$. Then for an isotropic observer $\boldsymbol{\gamma}(t)$ with $\dot{\boldsymbol{\gamma}} = \partial_t$ the energy density $\state{\varrho(q)}{\omega}$ for the minimally coupled quantized scalar field $\phi$ in an HIH state $\omega$, characterized via $T_k(t)$ and $\Xi(k)$ of theorem \ref{thm_homo_iso_states}, is given by
\begin{align}
\state{\varrho(q)}{\omega} = &\,\dfrac{1}{(2\pi)^3} \ainteg{B_{k_0}}{}{\varrho_k(t)}{\,^3 \mathbf{k}}+ \dfrac{1}{(2\pi)^3} \ainteg{\mathbb{R}^3\backslash B_{k_0}}{}{\left(\varrho_k(t)-\mathfrak{h}^\varrho(t,k)+\mathfrak{h}_{0}^\varrho(t)\right)}{\,^3\mathbf{k}}\nonumber\\
			&\qquad -\,\dfrac{k_0^4}{8\pi^2}\, \mathfrak{h}^\varrho_{1}(t) -\dfrac{1}{(2\pi)^3}\,\mathfrak{h}_{0}^\varrho(t)- \dfrac{k_0^2}{4\pi^2}\,\mathfrak{h}^\varrho_{-1}(t)\nonumber\\
			&\qquad + \dfrac{1}{2\pi^2} \left(1-\gamma-\log(k_0)\right)\mathfrak{h}^\varrho_{-3}(t)+ C_{00}(t),
\label{eqn_renorm_energy}
\end{align}
with the energy density per mode
\begin{equation}
\label{eqn_energydensity/preassure_mode}
\varrho_k(t) :=  \dfrac{1}{2}(1+2\Xi(k))\left(|\dot{T}_k(t)|^2+\omega_k(t)^2\,\abs{T_k(t)}^2\right)
\end{equation}
and $\omega_k(t)^2:= (\frac{k^2}{a(t)^2}+m^2)$. The energy density parametrix  $\mathfrak{h}^\varrho(t,k)$ is defined by (\ref{eqn_h_rho}) and $C_{0 0}$ by (\ref{eqn_renorm_ambiguity_rw_spacetime}). $B_{k_0}$ denotes the ball of radius $k_0$ around the origin and $\gamma$ the Euler-Mascheroni constant. 
\end{theorem} 
\begin{proof}
For an isotropic observer with $\dot{\boldsymbol{\gamma}}=\partial_t$ the energy density is just the EMT's 00-component, i.e. $\state{\varrho(q)}{\omega} = [(T^\text{can}_{00'}-1/3\,P\otimes \mathds{1})(\mathscr{W}^{\omega,\text{s}}_2-\mathcal{H}^\text{s}_{l,\varepsilon})]_\text{c}+C_{00}$. On the one hand the two-point distribution solves the KG equation, i.e. $(P\otimes \mathds{1})\mathscr{W}^{\omega,\text{s}}_2=0$, and on the other hand it was shown in \cite{moretti} that $[(P\otimes \mathds{1}) \mathcal{H}^\text{s}_{l,\varepsilon}]_\text{c}= -3\clim{v_1}/\pi^2$. Hence, we need to determine 
\begin{equation}
\label{eqn_energy_density1}
\state{\varrho}{\omega} = \clim{T^\text{can}_{00'}(\mathscr{W}^{\omega,\text{s}}_2-\mathcal{H}^\text{s}_{l,\varepsilon})} - \frac{1}{\pi^2}\, \clim{v_1} +C_{00}.
\end{equation}
We are going to take the coincidence limit along spacelike directions, i.e. we need to restrict the distributions $T^\text{can}_{00'}\mathscr{W}^{\omega,\text{s}}_2$ and $T^\text{can}_{00'}\mathcal{H}^\text{s}_{l,\varepsilon}$ to equal times $t=t'$. For that purpose let us define the embedding $\imath: \Sigma \rightarrow M, \mathbf{x} \mapsto q_\mathbf{x}:= (t,\mathbf{x})$. Associated to this embedding is the conormal bundle
\begin{equation*}
\begin{array}{r c l}
N^* \Sigma &:= & \left\{(q_\mathbf{x},\,\xi)\in T^*M\; :\; \mathbf{x}\in \mathbb{R}^3, \; \text{d}\imath^*(\xi) =0  \right\}\\[0.3cm]
	   & = & \left\{(q_\mathbf{x},\,\xi)\in T^*M\; :\; \mathbf{x}\in \mathbb{R}^3, \; \xi = \alpha\, \text{d}t \quad (\alpha \in \mathbb{R}) \right\}.
\end{array}
\end{equation*} 
Now, the pull-back $(\imath \otimes \imath)^* u$, i.e. the restriction of $u\in \mathcal{D}'(M\times M)$ to $\Sigma\times \Sigma$,  yields a well-defined distribution on $\Sigma \times \Sigma$ whenever the conormal bundle $N^*\Sigma \times N^*\Sigma$ does not intersect with the distribution's wavefront set\footnote{The wavefront set is a subset of the cotangent bundle $T^*M\times T^*M$ and encodes the position and direction of a distribution's singularities. For a rigorous definition see \cite{hoermander2}.} $WF(u)$, i.e. $(N^*\Sigma \times N^*\Sigma)\cap WF(u) = \emptyset$, cf. corollary 8.2.7 in \cite{hoermander2}. Furthermore, Fewster \& Smith have shown in \cite{fewster_smith} that all elements of $WF(T^\text{can}_{a b'}\,\mathscr{W}^{\omega,\text{s}}_2)$ and $WF(T^\text{can}_{a b'}\,\mathcal{H}^{\text{s}}_{l,\varepsilon})$ have to be lightlike. Since $N^*\Sigma$ only contains timelike elements, $(N^*\Sigma \times N^*\Sigma)\cap WF(u)$ is indeed empty for $u\in \{T^\text{can}_{a b'}\,\mathscr{W}^{\omega,\text{s}}_2, T^\text{can}_{a b'}\,\mathcal{H}^{\text{s}}_{l,\varepsilon}\}$ and the restrictions exist.\\
In the next step, we will deduce the explicit form of those two restricted distributions where we start with $[T^\text{can}_{a b'}\,\mathscr{W}^{\omega,\text{s}}_2]_{t'=t}$. Using theorem \ref{thm_homo_iso_states} the two-point distribution of a homogeneous and isotropic state can be rewritten for $f,h \in C^\infty_0(M)$ as 
\begin{equation*}
\mathscr{W}^\omega_2(f,h) = \lim\limits_{\delta \rightarrow 0} \int\limits_M \int\limits_M \int\limits_{\mathbb{R}^3} \mathit{w}_k(t,t')\,\Exp{\text{i}\mathbf{k(x-x')}-\delta\,k} \text{d}^3\mathbf{k}\;f(q) h(q') \text{d}\mu_g(q)\text{d}\mu_g(q')\,;
\end{equation*}
to obtain this equation from (\ref{eqn_2pd_his}) one first moves the Fourier factors $\Exp{\text{i}\mathbf{k}\,\cdot}$ from the test functions to $\mathit{w}_k(t,t')$, then includes the regularizing $\Exp{-\delta k}$ term and by virtue of Fubini's theorem exchanges the $\mathbf{k}$-integration with the $\mu_g(q)$- and $\mu_g(q')$-integration. Hence, the two-point function reads
\begin{equation}
\label{eqn_2pt_his_3}
\mathscr{W}^\omega_2(q,q') =\lim\limits_{\delta \rightarrow 0} \int\limits_{\mathbb{R}^3} \mathit{w}_k(t,t') \Exp{\text{i}\mathbf{k(x-x')}-\delta\,k}\text{d}^3 \mathbf{k},
\end{equation}  
where the limit is understood in the distributional sense. As the $\mathit{w}_k(t,t')$ are polynomially bounded we can simply differentiate underneath the integral in order to obtain $\nab{a}\otimes\nab{b'}\mathscr{W}^\omega_2(q,q')$ and thus find for the ``point-split energy density''
\begin{align*}
(T^\text{can}_{0 0'}\,\mathscr{W}^{\omega}_2)(q,q')  &= \lim_{\delta\rightarrow 0} \int\limits_{\mathbb{R}^3}\frac{1}{2}\left(\partial_t\partial_{t'} \mathit{w}_k(t,t')+(k^2 a(t)^{-2}+m^2)\,\mathit{w}_k(t,t')\right)\\[-0.25cm]
			& \hspace{6cm} \times \Exp{\text{i}\mathbf{k(x-x')}-\delta\,k}\text{d}^3\mathbf{k}.
\end{align*} 
Finally, after plugging in the definition of $ \mathit{w}_k(t,t')$, cf. (\ref{eqn_mode_hi_state}), we set $t=t'$, symmetrize with respect to $\mathbf{x}$ and $\mathbf{x'}$ and obtain the desired result
\begin{equation}
\label{eqn_unrenorm_energy_density}
[(T^\text{can}_{0 0'}\,\mathscr{W}^{\omega,\text{s}}_2)(q,q')]_{t'=t} = \lim\limits_{\delta \rightarrow 0}\,\frac{1}{(2\pi)^3}\, \integ{\mathbf{R}^3}{}{\varrho_k(t)\,\Exp{-\delta k} \cos(\mathbf{k(x-x')})}{^3\mathbf{k}}.
\end{equation} 
\\ 
Next we need to find a similar mode expansion for $[T^\text{can}_{0 0'}\,\mathcal{H}^{\text{s}}_{l,\varepsilon}]_{t'=t}$. Suppose, the kernel of a distribution $u\in \mathcal{D}'(\mathbb{R}^3\times\mathbb{R}^3)$ depends on $\mathbf{x,x'}\in \mathbb{R}^3$ solely via $z:=\norm{\mathbf{x-x'}},$ i.e. the distribution is invariant under the action of the Euclidean group. Then we can write the distribution $u$ for $f,h \in C^\infty_0(\Sigma,\mathbb{R})$ as   
\begin{align}
\label{eqn_TH_distribution}
u(f,h) & = \ainteg{\mathbb{R}^3}{}{\,\ainteg{\mathbb{R}^3}{}{u(\norm{\mathbf{x-x'}})f(\mathbf{x})h(\mathbf{x'})}{^3\mathbf{x}}}{^3\mathbf{x'}}  = \ainteg{\mathbb{R}^3}{}{u(z)\,(f * h)(\mathbf{z})}{^3\mathbf{z}},
\end{align}
where $(f*h)(\mathbf{z}):= \int_{\mathbb{R}^3}f(\mathbf{x}) h(\mathbf{x-z})\text{d}^3\mathbf{x}$ denotes the convolution of $f$ and $h$. Using the convolution theorem, which states $f*h=\mathcal{F}^{-1}[\fourier{f}\cdot\overline{\fourier{\overline{h}}}]$, and again the homogeneity and isotropy of $u$, which implies $\fourier{u}= (2\pi)^3\,\mathcal{F}^{-1}[u]$, equation (\ref{eqn_TH_distribution}) turns out to be equivalent to 
\begin{equation*}
u(f,h) = \frac{1}{(2\pi)^3}\integ{\mathbb{R}^3}{}{\fourier{u}(k)\,\fourier{f}(\mathbf{k})\overline{\fourier{\overline{h}}}(\mathbf{k})}{^3\mathbf{k}}.
\end{equation*}
This expression for the distribution $u$ is of the same type as the one for the two-point distribution of a homogeneous and isotropic state, cf. (\ref{eqn_2pd_his}). Hence, by repeating the same manipulations, that led to equation (\ref{eqn_2pt_his_3}), we find for the kernel associated to $u$:
\begin{align*}
u(\mathbf{x,x'}) & = \frac{1}{(2\pi)^3}\,\lim\limits_{\delta \rightarrow 0}\int\limits_{\mathbb{R}^3} \fourier{u}(k)\,\Exp{-\delta\,k}\!\cos(\mathbf{k(x-x')})\text{d}^3 \mathbf{k},
\end{align*}
where the $\delta$-limit is understood in the distributional sense. Since the flat FRW spacetime is spatially homogeneous and isotropic, the square geodesic distance and thus $[T^\text{can}_{00'}\mathcal{H}^\text{s}_{l,\varepsilon}]_{t'=t}$, too, depend on the spatial coordinates only via $\norm{\mathbf{x-x'}}$. Therefore, we replace in the previous equation the distribution $u$ by $[T^\text{can}_{00'}\mathcal{H}^\text{s}_{l,\varepsilon}]_{t'=t}$, combine it with equation (\ref{eqn_2pt_his_3}) and the definition of the energy density, cf. (\ref{eqn_energy_density1}), and finally obtain
\begin{align*}
\state{\varrho(q)}{\omega} & = \lim\limits_{\mathbf{x'}\rightarrow \mathbf{x}}\, \lim\limits_{\delta \rightarrow 0}\, \frac{1}{(2\pi)^3} \int\limits_{\mathbb{R}^3}\left(\varrho_k(t)-\fourier{[T^\text{can}_{0 0'}\,\mathcal{H}^{\text{s}}_{l,0}]^o_{t'=t}}(t,k)-\frac{[v_1]_c}{\pi^2}\delta_0(\mathbf{k})\right)\\[-0.2cm]
			& \hspace*{5cm}\times\,\Exp{-\delta k}\cos(\mathbf{k(x-x')})\text{d}^3\mathbf{k} + C_{00},
\end{align*}
where we set $\varepsilon =0$, since we are taking the coincidence limit along spacelike directions and thus no regularization for $\mathcal{H}^\text{s}_{l,\varepsilon}$ is needed. Now we apply lemma \ref{lem_h_rho/p} and as $\omega$ is a Hadamard state the function $(\varrho_k(t)-\mathfrak{h}^\varrho(t,k))$ is of order $k^{-5}$ for large $k$ so that we can drop the $\delta$-limit and eventually take the coincidence limit, which yields
\begin{equation*}
\begin{array}{r c l}
\state{\varrho(q)}{\omega} & = & \dfrac{1}{(2\pi)^3}\,\ainteg{\mathbb{R}^3}{}{\left(\varrho_k(t)-\mathfrak{h}^\varrho(t,k)\right)}{^3\mathbf{k}}+\,C_{00}.
\end{array}
\end{equation*}
It seems that this expression exhibits an infrared divergence due to the occurring factor $\mathfrak{h}^{\varrho}_{-3}(t)\,k^{-3}$ in $\mathfrak{h}^{\varrho}(t,k)$. However, as was already explained in \cite{eltz_gott}, if one properly traces this singular factor back to its definition as the Fourier transform of the locally integrable function $\log(z)$, one can indeed obtain a finite expression $-$ the explicit calculation leading to expression (\ref{eqn_renorm_energy}) can be found in \cite{zschoche}.
\end{proof}

\begin{remark}
(i) Note that another approach to obtain a mode decomposition for the quantized energy density of a free scalar field on a spatially flat FRW spacetime can be found in the PhD theses of Schlemmer \cite{schlemmer} and Degner \cite{degner}.\\
(ii) Whenever $\mathfrak{h}^{\varrho}_{-3}(t)\equiv 0$, we can choose without loss of generality $k_0 = 0$. In this case equation (\ref{eqn_renorm_energy}) simplifies again to
\begin{equation}
\state{\varrho(q)}{\omega} =  \dfrac{1}{(2\pi)^3}\,\ainteg{\mathbb{R}^3}{}{\left(\varrho_k(t)-\mathfrak{h}^\varrho(t,k)\right)}{^3\mathbf{k}}+\,C_{00}.
\end{equation}
\end{remark}

\begin{lemma}
\label{lem_HIH_pressure}
The energy-momentum tensor on a spatially flat FRW spacetime evaluated in an HIH state is of perfect fluid type with respect to the isotropic observer $\gamma(t) = \partial_t$, i.e. $\state{T^{a b}}{\omega} = (\state{\varrho}{\omega}+\state{p}{\omega})\,(\partial_t)^a\,(\partial_t)^b- \state{p}{\omega} \,g^{a b}$.\\
In particular, the expectation value of the pressure with respect to any HIH state $\omega$ is given by 
\begin{equation}
\label{eqn_renorm_pressure}
\state{p(t)}{\omega} = -\state{\varrho(t)}{\omega}-\frac{1}{3H(t)} \dif{t}\state{\varrho(t)}{\omega}.
\end{equation}
\end{lemma}
\begin{proof}
The proof is given in the appendix.
\end{proof}

\begin{remark}
The previous lemma is a necessary requirement in order to use $\state{T^{a b}}{\omega}$ as a source for the semi-classical Friedmann equations and can thus be seen as an a posteriori justification for the restriction to HIH states.     
\end{remark}

\section{The Chaplygin Gas Equation of State}
The major advantage of the CG EoS is that it entails in combination with the classical Friedmann equations (\ref{eqn_friedmann_equations}) a behavior for the scale factor of an FRW spacetime which is surprisingly close to the one observed in nature $-$ namely, it incorporates a Big Bang, a (dark) matter dominated phase and a dark energy era. In the previous section we described how to (covariantly) define the EMT for the quantized scalar field in an HIH state on \emph{all} spatially flat FRW spacetimes. In particular, for such states the EMT is of perfect fluid type, which makes it suitable as a source term for the semi-classical Friedmann equations. Usually, if one wants to solve these equations with the quantized EMT as the source one has to seek a scale factor $a(t)$ and mode functions $T_k(t)$ of an HIH state $\omega$ such that the temporal KG equation (\ref{eqn_tkge}), the normalization condition (\ref{eqn_normkge}) and the semi-classical Friedmann equations, where $\state{\varrho}{\omega}$ and $\state{p}{\omega}$ have to be replaced by their explicit expressions as derived in section \ref{sec_EMT_HIH_state}, are simultaneously fulfilled $-$ note that from this perspective the renormalization ambiguities play a crucial role since they determine the highest order of the occurring time derivatives of the scale factor and thus the ``nature'' of the equations of motion for $a(t)$ and $T_k(t)$, cf. \cite{dap_fred_pin, eltz_gott, pinamonti2}. \\
In our case the situation is however somewhat different. We are already \emph{given} the scale factor, namely $a_\text{cg}$, we already know which EoS is needed to obtain this scale factor from the (semi-classical) Friedmann equations and we have to \emph{check} whether there are HIH states for the quantized scalar field that obey such an EoS on the given spacetime or not $-$ in particular, we will view the renormalization ambiguities not as a possibility to modify the structure of the equations of motion but as a freedom in the definition of energy density and pressure itself.\\ 
Therefore, to answer the question whether the quantized scalar field is a matter model that leads to a solution which consistently solves the semi-classical Friedmann equations and the CG EoS one ought to proceed as follows:\\
\\
\begin{tabular}{l}
$\emph{(1)}\,$ \parbox[t]{10.7cm}{assume that there are HIH states $\omega$ which entail the CG EoS, $\state{p}{\omega}=-A\,\state{\varrho}{\omega}^{-1},$ $A>0,$ and solve the semi-classical version of the Friedmann equations (\ref{eqn_friedmann_equations})\\
	 $\rightarrow$ determine the scale factor $a(t)$ compatible with both assumptions}\\
\\
$\emph{(2)}\,$ \parbox[t]{10.7cm}{solve the KG equation (\ref{eqn_tkge}) and the normalization condition (\ref{eqn_normkge}) with $a(t)$ taken from step $\emph{(1)}$ $\rightarrow$ this yields the mode functions $T_k(t)$}\\ 
\\
$\emph{(3)}\,$ \parbox[t]{10.7cm}{construct from $T_k(t)\hspace{-0.03cm}$ all HIH states $\omega\hspace{-0.03cm}$ and the associated energy density and pressure, cf. (\ref{eqn_renorm_energy}) and (\ref{eqn_renorm_pressure}), respectively}\\
\\
$\emph{(4)}\,$ \parbox[t]{10.7cm}{identify those HIH states that are consistent with the CG EoS, i.e. find an HIH state $\omega$ such that $\state{p}{\omega} = -A\,\state{\varrho}{\omega}^{-1}, \, A>0$, or show that there are no such states at all.}\\
\\
\end{tabular}
Note that by accomplishing step $\emph{(4)}$ one has either found a solution to the semi-classical Friedmann equations or has constructed a contradiction which thereby shows that the quantized minimally coupled scalar field is not a matter model which allows for a CG EoS.\\[0.2cm]
While step $\emph{(1)}$ is easily accomplished, as the solution $a(t)=a_\text{cg}(t)$ was given by equation (\ref{eqn_cg_scalefactor}) above, step $\emph{(2)}$ poses a serious problem. Although the assumption of a homogeneous and isotropic state simplifies the matter a lot, the temporal KG equation  for this particular scale factor $a_\text{cg}(t)$ is not analytically solvable.\\
For this reason, one could first study the existence of HIH states consistent with the CG EoS on less complex spacetimes. On Minkowski spacetime, for instance, it is straight forward to check that in an HIH state the energy density has to be constant in time. Hence, an equation of state $\state{p(t)}{\omega} = f(\state{\varrho}{\omega})$ is only possible if $\state{p(t)}{\omega}$ is time-independent as well. Furthermore, for the Minkowski vacuum state $\omega_0$ we have by Lorentz invariance $c = -\state{p}{\omega_0} = \state{\varrho}{\omega_0}$, where $c$ is an arbitrary constant stemming from the renormalization ambiguities. That is, since the vacuum state's energy density and pressure are the lower limits for the energy density and pressure in any HIH state, we can conclude that a CG EoS on Minkowski spacetime is only possible if the energy density as well as the pressure are constant and $c$ is chosen to be positive $-$ if the renormalization is accomplished via normal ordering, we have $c=0$ and no CG EoS can be realized. In order to examine the influence of gravity to the existence of HIH states entailing the CG EoS we are going to consider curved backgrounds with non-positive, constant scalar curvature, $R=\text{const.}$ We chose those spacetimes basically for two reasons. The first one is that, at least for a certain mass parameter of the scalar field, these spacetimes allow for an explicit solution of the temporal KG equation. This solution additionally induces a Hadamard state, which will become the starting point of our analysis on those spacetimes. On the other hand the CG scale factor $a_\text{cg}$ converges with respect to the cosmological time to a certain de Sitter space exponentially fast. As the de Sitter space satisfies the condition $R=\text{const.}$ one might hope to obtain some knowledge about the late time behavior of the energy density and pressure by the study of the respective late time behavior in the $R=\text{const.}$-case. Note that at least for a self-interacting scalar field theory on de Sitter space, and presumably also on spacetimes asymptotically approaching de Sitter space exponentially fast, every generic state converges to the Bunch-Davies state. This is known as the ``quantum cosmic no hair theorem/conjecture,'' cf. \cite{hollands_1, hollands_2, mar_mor}. \\[0.2cm]
Since the answer to the question of the existence of HIH states inducing the  CG EoS requires the prior knowledge of all HIH states on a given FRW spacetime, we first have to describe an algorithm how those can be obtained. We divide this task into two parts. In the first part, subsection \ref{sec_Homogeneous and Isotropic Hadamard States}, we characterize the transformation behavior between HIH states in terms of a certain function class. In the second step, subsection \ref{sec_Construction_of_a_Hadamard_State_on_certain_spacetimes}, we present a way how to obtain a (particularly simple) HIH state on certain FRW spacetimes in the first place. And finally, in subsection \ref{sec_NoGoTheorem}, we accomplish the analysis announced above for the FRW spacetimes with constant scalar curvature $R$ and the scalar field of mass $m= \sqrt{-\frac{1}{6}R}$, show that there exists exactly one state which yields the CG EoS $-$ in this state, the energy density and pressure are constant similar to the Minkowski case $-$ and comment afterwards on the situation for the actual case of interest, the CG scale factor.

\subsection{Homogeneous and isotropic Hadamard states}
\label{sec_Homogeneous and Isotropic Hadamard States}
Before we can actually characterize every HIH state assuming one of those is already given we need to introduce the concept of essentially rapidly decaying functions. So, let us start with this section's crucial definition.
\begin{definition}
\label{def_erd}
Let $\Omega \subseteq \mathbb{R}^n$, $n\in \mathbb{N},$ and $\text{d}\mu$ the associated Lebesgue measure. A function $f: \Omega \rightarrow \mathbb{C}$ is called \emph{essentially rapidly decaying} (ERD) on $\Omega$ if and only if, for all $m\in \mathbb{N}_0$, the functions $x \mapsto \abs{x}^m\,f(x),$ $x\in \Omega$, are in $L^1(\Omega,\text{d}\mu)$. The set of essentially rapidly decaying functions on $\Omega$ will be denoted by $\text{ERD}(\Omega)$.
\end{definition}

\begin{remark}
Note that all rapidly decaying functions\footnote{A function $f$ is rapidly decaying if and only if for all $m\in \mathbb{N}_0$ there is a $C_m>0$ such that $\abs{f(k)} < \frac{C_m}{1+\abs{k}^m}$.} are ERD. The opposite is however not true as can be seen by the ``comb function'' $f_\text{C}: \mathbb{R}^+_0 \rightarrow \mathbb{R},$
\begin{equation}
\label{eqn_example_erd}
x \mapsto \sum_{n=1}^\infty\,\Exp{\frac{n}{2}}\, \chi_{[n-\Exp{-n},n+\Exp{-n}\!]}(x),
\end{equation}
with $\chi_I$ being the characteristic function. Although this function is unbounded it is still ERD since the width of each peak decreases sufficiently fast.
\end{remark}
Next we will prove some basic yet needed properties for functions of essentially rapid decay.
\begin{lemma}
\label{lem_prop_erd_1}
Let $\Omega \subseteq \mathbb{R}^n$, $n\in \mathbb{N}$, and $\text{d}\mu$ the Lebesgue measure. Then it holds:
\begin{longtable}{ll}
($\mathit{1}$) & \parbox[t]{10.5cm}{$(f: \Omega \rightarrow \mathbb{C})\,\in \text{ERD}(\Omega)$ $\!\quad \Longleftrightarrow\! \quad \forall m\in \mathbb{N}_0:\quad(1+\abs{\,\cdot\,})^m f \in L^1(\Omega,\text{d}\mu)$ } \\
\\[-0.15cm]
($\mathit{2}$) &\parbox[t]{10.5cm}{$(f: \Omega \rightarrow \mathbb{R}^+_0) \in \text{ERD}(\Omega) \quad \Longleftrightarrow \quad$ All moments of $f$ exist.}\\ 
\\[-0.15cm]
($\mathit{3}$) &\parbox[t]{10.5cm}{Let $p$ be a polynomially bounded function and $f\in \text{ERD}(\Omega)$ then $(p\cdot f)\in \text{ERD}(\Omega)$.}\\
\\[-0.15cm]
($\mathit{4}$) &\parbox[t]{10.5cm}{$\text{ERD}(\Omega)$ is a vector space.\\[0.1cm]
				 If $f,g\geq 0$ and $(f+g)\in \text{ERD}(\Omega)$ then $f$ and $g$ are ERD.}\\
\\[-0.15cm]
($\mathit{5}$) & \parbox[t]{10.5cm}{Let $f,g:\Omega \rightarrow \mathbb{C},\inf\limits_{x\in \Omega}\abs{f(x)} > 0 \;\text{and} \; f\!\cdot\! g \in \text{ERD}(\Omega) $ then $ g\! \in\!\text{ERD}(\Omega)$.}\\
\\[-0.15cm]
($\mathit{6}$) &\parbox[t]{10.5cm}{If $f^p\in \text{ERD}(\Omega)$ for some $p\geq 1$ then so is $f$.}
\end{longtable}
\end{lemma}
\begin{proof}
The proof is given in the appendix.
\end{proof}

\begin{remark} 
If $f$ is additionally assumed to be polynomially bounded then the converse statement of claim $\emph{(6)}$ is also true for any $p\in [1,\infty)$, i.e. if $f\in \text{ERD}(\Omega)$ and polynomially bounded so is $f^p$. The assumed polynomial boundedness is however necessary. For instance, one can easily check that although the function $f_\text{C}$, cf. (\ref{eqn_example_erd}), is ERD, this does not hold for $f^2$.
\end{remark}
Having collected the required basic properties of ERD functions we can now state this section's central technical result, which will be used afterwards for the characterization of Hadamard states on FRW spacetimes. It connects the smoothness of the Fourier transform to ``decaying properties'' and is a generalization of theorem 2.3.1 in \cite{lukacs}.
\begin{lemma}
\label{lem_smoothness_moments}
Let $\varphi:\mathbb{R}^n \rightarrow \mathbb{R}$ be a non-negative polynomially bounded function and $\omega(\mathbf{z}):=\lim\limits_{\varepsilon \rightarrow 0}\mathcal{F}_\varepsilon[\varphi](\mathbf{z})$ its regularized Fourier transform, i.e. 
\begin{equation}
\label{def_reg_fourier}
\mathcal{F}_\varepsilon[\varphi](\mathbf{z}):= \integ{\mathbb{R}^n}{}{\varphi(\mathbf{k}) \text{e}^{\text{i}\,\mathbf{k} \mathbf{z} -\varepsilon\,k}}{^{\,n}\mathbf{k}}
\end{equation}
with $k:=\norm{\mathbf{k}}$. Then $\omega \in C^\infty(\mathbb{R}^n)$ if and only if $\varphi$ is ERD. 
\end{lemma}

\begin{proof}
The proof is given in the appendix.
\end{proof}

With these technical tools at hand we will now turn to this section's main objective, namely the characterization of all HIH states on a spatially flat FRW spacetime $-$ the result presented here was obtained in joint work with Benjamin Eltzner\footnote{Private communication, permission for publication granted. Also see \cite{eltzner_phd}.} and Nicola Pinamonti\footnote{Private communication, permission for publication granted. Note that in \cite{pinamonti2} it was claimed that the function $\alpha(k)$, as defined in theorem \ref{thm_characterisation_hihs}, has to be rapidly decaying. This is however wrong since the correct function class turns out to be rather $\text{ERD}(\mathbb{R}^3)$.}.
\begin{theorem}
\label{thm_characterisation_hihs}
Let $(M,g)$ be a spatially flat FRW spacetime with scale factor $a\in C^\infty(\tilde{I},\mathbb{R}^+)$, $\tilde{I}\subset \mathbb{R}$ being an open interval. Furthermore, suppose that $\omega$ is a quasifree homogeneous and isotropic state on $(M,g)$ and $\omega_0$ a pure HIH state $-$ the associated $k$-polynomially bounded mode functions of $\omega_0$ are denoted by $q_k(t)$ and the two-point function by $\mathscr{W}^{\omega_0}_2$, cf. theorem \ref{thm_homo_iso_states}.\\
Then $\omega$ is a Hadamard state if and only if its associated two-point function $\mathscr{W}_2^{\omega}$ can be written as 
\begin{align}
\mathscr{W}_2^{\omega}(q,q') & =\lim\limits_{\varepsilon \rightarrow 0}\, \frac{1}{(2\pi)^3}\, \int\limits_{\mathbb{R}^3}\left(\Xi(k)\,T_k(t)\bar{T}_k(t')+\left(1+\Xi(k)\right)\bar{T}_k(t)T_k(t')\right)\nonumber\\[-0.3cm]
&\hspace{7cm} \times \Exp{\text{\emph{i}}\,\mathbf{k}(\mathbf{x}-\mathbf{x}')-\varepsilon k}\text{d}^3\mathbf{k}\nonumber\\[0.1cm]
\label{eqn_characterisation_hihs_2pf}
\text{with} & \quad T_k(t):= \sqrt{1+\alpha(k)^2}\,q_k(t)+\alpha(k)\,\Exp{-\text{\emph{i}}\, \psi(k)}\,\bar{q}_k(t),
\end{align} 
where $\Xi, \alpha \in \text{ERD}(\mathbb{R}^+_0)$ are real-valued, non-negative as well as polynomially bounded and $\psi$ is an integrable real-valued function.
\end{theorem}

\begin{proof}
By theorem \ref{thm_homo_iso_states} the two-point function of every homogeneous and isotropic quasifree state on an FRW spacetime is of the form 
\begin{align*}
\mathscr{W}_2^{\omega}(q,q')  = & \lim_{\varepsilon \rightarrow 0} \frac{1}{(2\pi)^3}\, \int\limits_{\mathbb{R}^3}\left(\Xi(k)\,T_k(t)\bar{T}_k(t')+\left(1+\Xi(k)\right)\bar{T}_k(t)T_k(t')\right)\\[-0.3cm]
		& \hspace{6.8cm} \times \Exp{\text{i}\,\mathbf{k}(\mathbf{x}-\mathbf{x}') -\varepsilon\,k} \text{d}^3\mathbf{k}
\end{align*}
with $T_k(t)$ being a solution of the temporal KG equation (\ref{eqn_tkge}) and the normalization condition (\ref{eqn_normkge}). If there is another solution of these two equations, in our case $q_k$, $T_k$ can be expressed in terms of $q_k$ by a Bogoliubov transformation, cf. remark \ref{rem_bogoliubov_trafo}, namely we have
\begin{equation}
\label{eqn_proof_bogoliubov_trafo_chara_his}
T_k(t) = \lambda(k)\,q_k(t) + \mu(k)\,\bar{q}_k(t).
\end{equation}
Since $q_k$ represents by definition a pure Hadamard state, it is on the one hand polynomially bounded in $k$ and on the other hand of order $k^{-\frac{1}{2}}$ for large $k$ $-$ the latter property follows from the observation that the two-point function's mode decomposition, which is basically given by $w_k(t,t') = (2\pi)^{-3}\,\bar{q}_k(t)q_k(t')$, cf. (\ref{eqn_2pt_his_3}), has to be of the same leading order as the mode decomposition of the Hadamard parametrix, which can be checked to be of order $k^{-1}$ for large $k$. Therefore, since $T_k$ is polynomially bounded in $k$ as well, $\mu(k)$ and $\lambda(k)$ have to be polynomially bounded in $k$, too.\\[0.3cm]
Now, let $t^*\in \tilde{I}$ be fixed. Since $\tilde{I}$ is open there exists a $t_*\in \tilde{I}$ such that $t_*<t^*$ and we can define the scale factor $a^* \in C^\infty(\tilde{I},\mathbb{R}^+)$ by 
\begin{equation}
a^*(t) = \left\{
\begin{array}{c l c}
a(t)		 & \text{for} & t^*<t\\[0.1cm]
a_\text{i}(t) & \text{for} & t_*\leq t \leq t^*\\[0.1cm]
1 		 & \text{for} & t<t_*,
\end{array}
\right.
\end{equation}
where $a_\text{i}$ is a function smoothly interpolating between the values $1$ and $a(t^*)$. Obviously, our initial solutions $T_k$ and $q_k$ also have to solve the temporal KG equation with the scale factor $a^*$ for $t > t^*$. Therefore, each of them extends to a solution of the ``deformed scale factor KG equation with $a^*$'' for all $t \in \tilde{I}$, in the following denoted by $T^*_k$ and $q_k^*$, respectively. The same holds for their associated two-point function $\mathscr{W}^{*,\omega}_2$ and $\mathscr{W}^{*,\omega_0}_2$, respectively. That is, we have $T^*_k = \lambda\, q^*_k + \mu\, \bar{q}_k^*$ where $\left.T^*_k\right|_{t > t^*} = T_k$ and $\left.\mathscr{W}^{*,\omega}_2\right|_{t > t^*} = \mathscr{W}^{\omega}_2$. It is important to note that this deformation of the scale factor does not affect the Bogoliubov coefficients $\lambda, \mu$ as they could for instance be uniquely determined from $T_k, q_k, \dot{T}_k$ and $\dot{q}_k$ at the instant, say, $t_0>t^*$.\\[0.3cm]
There is however yet another distinguished Hadamard state for the deformed spacetime. Namely, at early times the deformed spacetime equals the Minkowski spacetime and thus we can define for $t<t_*$ the ``vacuum state'' $\omega_\infty$ with the mode function $v_k(t):=\frac{1}{\sqrt{2\omega_k}}\Exp{\text{i} \omega_k\,t}$ where $\omega_k^2:=\frac{k^2}{a^*(t)^2}+m^2$. This mode function is again extended to a solution of the complete spacetime, say $v^*_k(t)$, thus yielding the Hadamard state $\omega^*_\infty$ with two-point function $\mathscr{W}^{*,\omega_\infty}_2$. The Bogoliubov transformations mapping $\omega_\infty$ into $\omega$ and $\omega_0$ are given by
\begin{equation*}
T^*_k(t) = \gamma_T(k)\,v^*_k(t) + \delta_T(k)\,\bar{v}^{\,*}_k(t) \quad \text{and}\quad q^*_k(t) = \gamma_q(k)\,v^*_k(t) + \delta_q(k)\,\bar{v}^{\,*}_k(t),
\end{equation*}
respectively, such that $\abs{\gamma_.(k)}^2-\abs{\delta_.(k)}^2=1$ holds.\\

Now we consider the difference of $\mathscr{W}^{*,\omega}_2-\mathscr{W}^{*,\omega_\infty}_2$, which is given at $(p,p')\in M\times M$ by
\begin{align*}
\left(\mathscr{W}^{*,\omega}_2-\mathscr{W}^{*,\omega_0}_2\right)(p,p') & = \frac{1}{4\pi^3}\,\lim_{\varepsilon\rightarrow 0}\, \int\limits_{\mathbb{R}^3}\! \left(1 + 2 \Xi(k)\right) \mathfrak{Re}\!\left[\gamma_T(k)\, \bar{\delta}_T(k)\, v^*_k(t)\, v_k^*(t')\right]\\[-0.2cm]
		&\hspace{5.1cm }\times\,\Exp{\text{i}\,\mathbf{k}(\mathbf{x}-\mathbf{x}') -\varepsilon\,k}\text{d}^3\mathbf{k}\nonumber \\
		& \quad+ \frac{1}{4\pi^3}\, \lim_{\varepsilon\rightarrow 0}\, \int\limits_{\mathbb{R}^3}\left(\Xi(k) + 2 \Xi(k) \abs{\delta_T(k)}^2+\abs{\delta_T(k)}^2\right)\nonumber\\[-0.2cm]
		& \hspace{2.8cm}\times \mathfrak{Re}\!\left[v^*_k(t) \bar{v}^{\,*}_k(t')\right]\Exp{\text{i}\,\mathbf{k}(\mathbf{x}-\mathbf{x}') -\varepsilon\,k}\,\text{d}^3\mathbf{k}.
\end{align*}
Since $\omega$ and $\omega_\infty$ are both Hadamard states the difference of their two-point function has to be smooth. Therefore, by applying the operator 
\begin{equation*}
\mathscr{E}:= \partial_t\otimes \partial_{t'} + \frac{1}{a(t)a(t')}\sum^3_{i=1} \partial_i\otimes \partial_{i'} +m^2\,\mathds{1}\otimes \mathds{1}
\end{equation*} 
to the distribution $(\mathscr{W}^{*,\omega}_2-\mathscr{W}^{*,\omega_\infty}_2)$ we still obtain a smooth function for the kernel of $\mathscr{E}(\mathscr{W}^{*,\omega}_2-\mathscr{W}^{*,\omega_\infty}_2)$. This smooth kernel can be calculated in the same way as described in the proof of theorem \ref{thm_renorm_energy/pressure}. After taking the temporal coincidence limit $t'\rightarrow t$ as well as setting $\mathbf{z}:=\mathbf{x}-\mathbf{x'}$ and restricting to times $t<t_*$, where the mode function $v^*_k(t)$ was given by the plane wave $\frac{1}{\sqrt{2\omega_k}}\Exp{\text{i} \omega_kt}$, the result of this calculation reads for $t<t_*$
\begin{align*}
\left[\mathscr{E}(\mathscr{W}^{*,\omega}_2-\mathscr{W}^{*,\omega_\infty}_2)\right]_{t'=t} & = \lim_{\varepsilon\rightarrow 0}\, \int\limits_{\mathbb{R}^3}\frac{\omega_k}{4\pi^3} \left(\Xi(k) + 2 \Xi(k)\abs{\delta_T(k)}^2+\abs{\delta_T(k)}^2\right)\\[-0.25cm]
&\hspace{5.5cm} \times \Exp{\text{i}\,\mathbf{k} \mathbf{z} -\varepsilon\,k} \text{d}^3\mathbf{k}\\[0.15cm]
		& = \lim_{\varepsilon \rightarrow 0} \mathcal{F}_\varepsilon\! \left[\frac{\omega_k}{4\pi^3}\! \left(\Xi(k)\! +\! 2\Xi(k)\! \abs{\delta_T(k)}^2\!+\!\abs{\delta_T(k)}^2\right)\right]\!(\mathbf{z}),
\end{align*} 
where $\mathcal{F}_\varepsilon$ denotes the regularized Fourier transform (\ref{def_reg_fourier}). Again, since $\omega$ and $\omega_\infty$ are both Hadamard states the right hand side of the previous equation actually has to be a smooth function on $\mathbb{R}^3$. But since $\omega_k\cdot(\Xi + 2\abs{\delta_T}^2 \Xi +\abs{\delta_T}^2)\geq 0$ this statement is by lemma \ref{lem_smoothness_moments} equivalent to the conclusion that $\omega_k\cdot(\Xi + 2\abs{\delta_T}^2 \Xi+\abs{\delta_T}^2)$ is ERD. In particular, by employing lemma \ref{lem_prop_erd_1}.\emph{(4)}/\emph{(5)}, this means that $\Xi$ and $\abs{\delta_T}^2$ have to be ERD and from lemma \ref{lem_prop_erd_1}.\emph{(6)} it follows that this already has to hold for $\abs{\delta_T}$ itself. Replacing in the above argumentation $\omega$ by $\omega_0$ we immediately obtain the ERD property for $\delta_q$, too.\\[0.25cm]
Now we have to show that the Bogoliubov coefficient $\mu$ connecting $\omega$ and $\omega_0$, cf. equation (\ref{eqn_proof_bogoliubov_trafo_chara_his}), is also ERD. In order to do so, we first express $v^*_k$ and $\bar{v}^{\,*}_k$ in terms of $q^*_k$ and $\bar{q}^*_k$, then replace $T^*_k$ in (\ref{eqn_proof_bogoliubov_trafo_chara_his}) with $T^*_k = \gamma_T\,v^*_k + \delta_T\,\bar{v}^*_k$, substitute $v^*_k$ and $\bar{v}^{\,*}_k$ afterwards and finally read off $\lambda$ and $\mu$. That is, using $\abs{\gamma_q}^2-\abs{\delta_q}^2=1$, we easily find
\begin{equation*}
v^*_k = \bar{\gamma}_q\,q^*_k-\delta_q\,\bar{q}^*_k \quad \Longrightarrow \quad \gamma_T\,(\bar{\gamma}_q\,q^*_k-\delta_q\,\bar{q}^*_k)+ \delta_T\,(\gamma_q\,\overline{q}^*_k-\bar{\delta}_q\,q^*_k)= \lambda\,q^*_k + \mu\,\bar{q}^*_k.
\end{equation*}
If we collect on the right equation's left hand side all terms proportional to $\overline{q}^*_k$ we see that $\mu = \gamma_T \delta_q-\gamma_q \delta_T$. Now $\gamma_q$ and $\gamma_T$ are polynomially bounded which implies, via $\delta_q, \delta_T \in \text{ERD}(\mathbb{R}^+_0)$ and lemma \ref{lem_prop_erd_1}.\emph{(3)}, $\mu \in \text{ERD}(\mathbb{R}_0^+)$.\\[0.25cm]
Finally, splitting $\lambda$ and $\mu$ into their modulus and phase, where we choose without loss of generality $\lambda$ to be real, such that the normalization condition $\abs{\lambda}^2-\abs{\mu}^2=1$ is fulfilled, we obtain the form for $\lambda$ and $\mu$ as claimed in the above theorem.  
\end{proof}

\subsection{Construction of a Hadamard State on certain Spacetimes}
\label{sec_Construction_of_a_Hadamard_State_on_certain_spacetimes}

Theorem \ref{thm_characterisation_hihs} allows to compute any HIH state on a spatially flat FRW spacetime as soon as one HIH state is known. However, to find an initial Hadamard state, that can serve as a, so to speak, \emph{seed state} for the above construction, is not at all an easy task. Nevertheless, there are some results how to obtain Hadamard states on quite general classes of FRW spacetimes $-$ for instance the ``bulk-to-boundary states'' by Dappiaggi, Moretti and Pinamonti \cite{dap_mor_pin_2} or Olbermann's states of low energy \cite{olbermann}. Here we want to present yet another approach based on the results of Junker and Schrohe \cite{junker_schrohe} and Olbermann \cite{olbermann}. Namely, we use the singular part of the Hadamard parametrix to construct the seed state. Although this approach is limited to FRW spacetimes whose scale factor fulfills a certain ODE the Hadamard state obtained in this way turns out to be of a rather simple form.\\[0.25cm]
We will proceed as follows. First we will state the ODE the scale factor has to obey and then define the associated seed solution of the temporal KG equation, which gives rise to a seed state. After providing a short list of three examples we prove the Hadamard property.\\
\begin{lemma}
\label{lem_gauge_state}
Let $a(t)$ be the scale factor of a spatially flat FRW spacetime and $\mathfrak{s}(t,k):= \mathfrak{h}^\varrho(t,k)-\mathfrak{h}^\varrho_0(t)$ the $k$-dependent part of the Hadamard mode expansion (\ref{eqn_h_rho}).  Using the abbreviation $\partial_t f \equiv \dot{f}$, we can define the function
\begin{equation}
\label{def_j}
\jmath(t):=\frac{6H(t) \mathfrak{s}(t,k)+\dot{\mathfrak{s}}(t,k)}{H(t) (m^2+2\omega_k(t)^2)}\,=\,2\,\dfrac{\partial_t\left(a(t)^6\,\mathfrak{s}(t,k)\right)}{\partial_t\left(a(t)^6\,\omega_k(t)^2\right)}.
\end{equation}
Suppose this $\jmath(t)$ is a solution to the differential equation 
\begin{equation}
\label{eqn_scalefactor_cond}
2\jmath(t) \left( 2\,\omega_k(t)^2\jmath(t)+3H(t) \dot{\jmath}(t)+\ddot{\jmath}(t) \right) = \frac{c_k}{a(t)^6}+\dot{\jmath}(t)^2
\end{equation}
for some arbitrary constant $c_k\neq 0$ $-$ in the following we will refer to equation (\ref{eqn_scalefactor_cond}) as the \emph{consistency equation}. Then setting for
\begin{align}
\label{def_q(c>0)}
c_k>0: &\,  q(t):= \sqrt{\frac{\jmath(t)}{\sqrt{c_k}}}\; e^{\text{i}\,\theta(t)},\quad \dot{\theta}(t) = \frac{\sqrt{c_k}}{2\,\jmath(t)\,a(t)^3},\\
\nonumber\\
\label{def_q(c<0)}
c_k<0: &\,  q(t):= \sqrt{\frac{\jmath(t)}{\sqrt{|c_k|}}}\; \left( \cosh(\theta(t))+\text{i}\,\sinh(\theta(t))  \right),\quad \dot{\theta}(t) = \frac{\sqrt{|c_k|}}{2\,\jmath(t)\,a(t)^3},
\end{align}
the function $q(t)$ solves the temporal Klein-Gordon equation (\ref{eqn_tkge}) and the normalization condition (\ref{eqn_normkge}). In the following we will call this $q(t)$ the \emph{seed solution}.
\end{lemma}
\begin{proof}
We start with the case $c_k>0$. By differentiation of $q(t)$ with respect to $t$ and the subsequent substitution of $\dot{\theta}(t)$  we find
\begin{align}
& \dot{q}(t) = \frac{q(t)}{\jmath(t)} \,\left( \frac{\dot{\jmath}(t)}{2} + \frac{\text{i}\,\sqrt{c_k}}{2\,a(t)^3}\right) \nonumber\\[0.2cm]
& \ddot{q}(t) = -\frac{q(t)}{4\jmath(t)^2 a(t)^6}\,\left(c_k+6\text{i}\sqrt{c_k}\,H(t)\, \jmath(t)\, a(t)^3+\left(\dot{\jmath}(t)^2 -2\jmath(t)\, \ddot{\jmath}(t)\right) a(t)^6\right)\nonumber. 
\end{align}
After inserting these in the temporal Klein-Gordon equation we get the desired result:
\begin{align}
\ddot{q}(t)\,& +\,3H(t)\, \dot{q}(t)+\omega_k(t)^2 q(t)\nonumber\\
			  & = \frac{q(t)}{4\jmath(t)^2} \left(-\frac{c_k}{a(t)^6}-\dot{\jmath}(t)^2 +2\jmath(t)\,\ddot{\jmath}(t)+6H(t)\,\jmath(t)\,\dot{\jmath}(t) +4\jmath(t)^2\, \omega_k(t)^2\right)\nonumber\\
			  & = 0.\nonumber
\end{align}
In the same way we verify the normalization condition:
\begin{equation*}
\bar{q}(t)\,\dot{q}(t)-q(t)\,\dot{\bar{q}}(t) = \frac{1}{\sqrt{c_k}}\,\left( \frac{\dot{\jmath}(t)}{2} + \frac{\text{i}\,\sqrt{c_k}}{2\,a(t)^3}\right)-\frac{1}{\sqrt{c_k}}\,\left( \frac{\dot{\jmath}(t)}{2} + \frac{\text{i}\,\sqrt{c_k}}{2\,a(t)^3}\right) = \frac{\text{i}}{a(t)^3}.
\end{equation*}
The  remaining case can be proven analogously.
\end{proof}
\begin{remark}
\label{rem_cons_eqn}
$(i)$ The function $\jmath(t;k)$ is by definition a rational function in $k$ with coefficients depending on $a^{(n)}(t)$, $n=0,...,4$. Since $a(t)$ is independent of $k$, $c_k$ thus has to be a rational function in $k$, too. Therefore once $c_k$ is given one can simplify the consistency equation by first recasting this equation into the form $\sum^l_{n=0} F_{n}[a(t), ... , a^{(6)}(t); m]\,k^{n} = 0,\,l\in \mathbb{N}$, and then solving each $k$-order separately. For instance, if one assumes that $\partial_k c_k =0$, i.e. $c_k$ is just an ordinary constant, the three highest $k$-orders are given by $F_{14} \propto (c_k-1),\, F_{12} \propto(c_k-1)\, m^2\, a(t)^2$ and $F_{10} \propto(c_k-1)\, m^4\, a(t)^4$ so that in this case actually $c_k=1$ has to hold.\\
\\
$(ii)$ Note that for a seed solution with $c_k>0$ the energy density per mode $\varrho_k(t)$ is basically given in terms of $a(t)$ and $\mathfrak{s}(t,k)$, namely one finds
\begin{equation}
\label{eqn_rho_ck>0}
\varrho_k(t) = \left(\varrho_k(t_0)-\dfrac{1}{\sqrt{c_k}}\,\mathfrak{s}(t_0,k)\right)\left(\frac{a(t_0)}{a(t)} \right)^6 +\dfrac{1}{\sqrt{c_k}}\,\mathfrak{s}(t,k).
\end{equation}
\end{remark}
We are now providing (an incomplete list of) three examples of scale factors that, as can be checked, fulfill the ODE (\ref{eqn_scalefactor_cond}).\\
\\
\emph{Example 1: $c_k=1$, $m=0$, $\varrho_k(t_0)= \mathfrak{s}(t_0,k)$.} All spacetimes leading to a solution of the consistency equation with such a parameter configuration solve one of the following four differential equations\footnote{These ODEs were derived by the ``algorithm'' we outlined in remark \ref{rem_cons_eqn}.$(i)$, i.e. by solving the consistency equation (\ref{eqn_scalefactor_cond}) and $\varrho_k(t_0)=\mathfrak{s}(t_0,k)$ in each order of $k$ setting $m=0$.}:
\begin{align}
\label{eqn_c1b0m0_1}
\pm \sqrt{2}\,\left(-\dot{H}\right)^{\frac{3}{2}} + 2 H\,\dot{H}+\ddot{H} &= 0\\
\label{eqn_c1b0m0_2}
\pm \sqrt{2}\,\left(2H^2+\dot{H}\right)^{\frac{3}{2}} + 4 H^3+6H\,\dot{H}+\ddot{H} &= 0.
\end{align}
Some analytic solutions of these two ODEs are given by
\begin{itemize}
\item[$\cdot$] $a(t) = (\beta\pm\alpha\, t)^n \qquad \left(n\in \left\{0,\frac{1}{2},\frac{2}{3},2\right\}\right)$\\[-0.3cm]
\item[$\cdot$] $a(t) = \beta\, \exp(\alpha\,t)$
\end{itemize}
where $\alpha$ and $\beta$ are arbitrary real constants.\\
Note that, as soon as $\mathfrak{h}^\varrho_{-3}(t) \neq 0$ the energy density $\varrho_k(t)$, cf. equation (\ref{eqn_rho_ck>0}), exhibits an non-integrable infrared divergence for $k\rightarrow 0$. The reason for this divergence is that we constructed the mode function $q(t)$ from the mode expansion of the Hadamard series which, in our ansatz, is not well-defined as a function at $k=0$ itself. To make this clearer let us explicitly write down the seed solution associated to the scale factor $a(t) = t^{\frac{2}{3}}$, i.e. 
\begin{equation*}
q(t) = \sqrt{\frac{1 + 9 k^2\, t^{\frac{2}{3}}}{k^3\,t^2}} \,\Exp{\text{i}\,\theta(t)} \quad \text{and}\quad \theta(t) = 3 k\,t^\frac{1}{3} - \arctan\left(3 k\,t^\frac{1}{3}\right).
\end{equation*} 
We see that $q(t)^2$ does not yield an intergable function in $k$ at $k=0$ as would be necessary for the definition of a state. Nevertheless, we can still regularize this solution for $0\leq k <k_0$, $k_0\in \mathbb{R}^+$, by applying a Bogoliubov transformation to $q(t)$ such that the new solution becomes regular at $k=0$. Such a regularizing transformation is for instance obtained from the (non unique) choice $\lambda = \frac{1}{2}\,(k^{2}+k^{-2})$ and $\mu = \frac{1}{2}\,(k^2-k^{-2})$ so that the completely regular solution $q_\text{r}$ reads
\begin{equation}
q_\text{r}(t) = \left\{ 
\begin{array}{c l}
\frac{1}{2}\,(k^{2}+k^{-2})\,q(t) + \frac{1}{2}\,(k^{2}-k^{-2})\,\bar{q}(t), & \text{for  } 0<k\leq k_0 \\
\\
q(t), & \text{for  } k_0<k < \infty.
\end{array}\right.
\end{equation}
\\[0.1cm]
\emph{Example 2: $R=\text{\emph{const}.}\leq 0,\,m=\sqrt{-\frac{1}{6}\,R}$}; each FRW spacetime with a constant non-positive scalar curvature belongs to one of the following five cases:\\[-0.3cm]
\begin{align*}
&m=0: \quad
a(t)= \left\{
\begin{array}{l}
 \beta_2\,\sqrt{\pm t + \beta_1}\\[0.2cm]
 \beta_2
\end{array}
\right.,\\[0.2cm]
& m>0:\quad a(t) = \left\{ 
\begin{array}{l}
 \beta_2\,\Exp{\pm\,\frac{m}{\sqrt{t}}}\\[0.2cm]
 \beta_2\,\sqrt{\cosh(\pm\,\sqrt{2}m\,t+\beta_1)}\\[0.2cm]
 \beta_2\,\sqrt{\sinh(\pm\,\sqrt{2}m\,t+\beta_1)}
\end{array}
\right. ,
\end{align*}
with $\beta_1 \in \mathbb{R}$ and $\beta_2>0$. Those scale factors are solutions to the consistency equation (\ref{eqn_scalefactor_cond}) if $c_k = 1$. Also in these cases we have $\varrho_k(t_0)-\mathfrak{s}(t_0,k)=0$ and $\mathfrak{h}_{-3}^\varrho(t)=0$, i.e. the associated $q(t)$ defines a proper state, which for the case of de Sitter spacetime equals the Bunch-Davies state.\\
\\
\emph{Example 3: }$a(t)\,= \alpha\,\sinh^\frac{2}{3}\left(\beta\pm \frac{3m}{2\sqrt{2}}\,t\right), m>0.$ If we choose $c_k = 1-\frac{(\alpha\,m)^6}{512k^6}$, i.e. $c_k$ changes its sign at $k_0 = \frac{\alpha\,m}{2\sqrt{2}}$, then $a(t)$ determines a solution of the consistency equation. Furthermore, one can check that the seed solution $q(t)$ associated to $a(t)$ does not exhibit any infrared divergences but an integrable pole of order $\abs{k-k_0}^{-\frac{1}{4}}$ at $k=k_0$. Finally, we note that this scale factor describes a dust filled universe with a positive cosmological constant \cite{carroll}.\\[0.2cm]
Now we can show that the states associated to the seed solution are Hadamard states. Here we will use an argument that was given by Olbermann in \cite{olbermann}. Namely, if a pure state can be approximated by an adiabatic vacuum state of arbitrarily high order, this state has to be a Hadamard state. So let us begin with the proper definition of the adiabatic vacuum states.
\begin{definition}
\label{def_adiabatic_state}
Let $T^{(n)}_k(t)$ be the mode function associated to a pure homogeneous and isotropic state $\omega$ on a spatially flat FRW spacetime with scale factor $a$. Suppose furthermore that $I=[T_0,T_1]$ is a compact interval. We define for $t_0, t\in I$ the function
\begin{equation}
\label{def_adiabatic_state_Wn}
W_{(n)}(t;k) := \frac{1}{\sqrt{2a(t)^3\,\Omega_{(n)}(t;k)}}\,\exp\left(\text{i}\,\ainteg{t_0}{t}{\Omega_{(n)}(y;k)}{y}\right)
\end{equation}  
where $\Omega_{(n)}$ is given iteratively by 
\begin{align}
\Omega_{(n+1)}(t;k)^2 &:= \omega_k(t)^2+\dfrac{3}{4}\,H(t)^2+\dfrac{1}{4}\,R(t)+\dfrac{3}{4}\left(\dfrac{\dot{\Omega}_{(n)}(t;k)}{\Omega_{(n)}(t;k)}\right)^2-\dfrac{1}{2}\,\dfrac{\ddot{\Omega}_{(n)}(t;k)}{\Omega_{(n)}(t;k)}\nonumber\\
\label{eqn_adiabatic_omega}
\Omega_{(0)}(t;k)   &:= \omega_k(t).
\end{align}
Then the state $\omega$ is called an \emph{adiabatic vacuum state of order $n$} if and only if its mode function $T^{(n)}_k$ satisfies for the initial time  $t_1 \in I$ the initial conditions
\begin{equation}
T^{(n)}_k(t_1) = W_{(n)}(t_1;k) \quad \text{and} \quad \dot{T}^{(n)}_k(t_1) = \dot{W}_{(n)}(t_1;k).
\end{equation}
\end{definition}
Note that it is a priori not clear that the $\Omega_{(n)}$ exist. Since the scale factor $a(t)$ is an arbitrary strictly positive function, the right hand side of (\ref{eqn_adiabatic_omega}) could in principle become negative. However, L\"uders and Roberts have shown that for all $n\in \mathbb{N}$ there is a number $\chi_n(I)$, $I$ being a closed time interval, such that $\Omega_{(n)}$ is always positive on $\mathcal{R}_n(I):=\{(t,k)\,|\,t\in I\,,\;k>\chi_n(I) \}$ and has continuous derivatives with respect to $t$, cf. \cite{lued_rob}. For $k\leq \chi_n(I)$ one can extrapolate the mode function $T^{(n)}_k$, similarly as in Example $\emph{1}$ above, in order to obtain an adiabatic vacuum state. For further properties of adiabatic vacuum states we refer the reader to \cite{junker_schrohe, lued_rob, olbermann}.\\ 

\begin{theorem}[Olbermann \cite{olbermann}]
\label{thm_olbermann_2}
Let $I=[T_0,T_1]\subset \mathbb{R}$ be a compact time interval and suppose $q_k:I\rightarrow \mathbb{C}$ are the mode functions of a pure homogeneous and isotropic quasifree state on a spatially flat FRW spacetime $(M,g)$ with the two-point distribution $\mathscr{W}^\omega_2$. Furthermore, let $T^{(n)}_k: I\rightarrow \mathbb{C}$ be the mode functions associated to a quasifree adiabatic vacuum state of order $n$ with two-point distribution $\mathscr{W}^{\omega_n}_2$, that is connected to $q_k(t)$ via the Bogoliubov transformation $q_k(t) = \alpha_n(k)\,T^{\,(n)}_k(t)+\beta_n(k)\,\bar{T}^{\,(n)}_k(t)$. If, for all $n$ and $k\rightarrow \infty$, $\abs{\alpha_n(k)}-1 = \mathcal{O}(k^{-2n-2})$ and $\abs{\beta_n(k)} = \mathcal{O}(k^{-2n-2})$ then $q_k(t)$ define a Hadamard state for all $t\in I$.
\end{theorem}
In order to show that the seed states are Hadamard states it therefore suffices to show that these states can be approximated by adiabatic vacuum states of arbitrary large order. Before we are able to make this last step we need two more preparatory lemmas.
\begin{lemma}
\label{lem_proof_hada_prop1}
Let $a(t)$ be a solution to the consistency equation (\ref{eqn_scalefactor_cond}). Then $\frac{\sqrt{c_k}}{2a(t)^3\jmath(t)}$ shows the asymptotic behaviour
\begin{equation}
\label{eqn_asymp_c/2aj}
\frac{\sqrt{c_k}}{2a(t)^3\jmath(t)} = \frac{k}{a(t)}+\frac{a(t)}{12k}\,(6m^2+R(t))+\mathcal{O}(k^{-3}) \quad \text{for}\quad k\rightarrow \infty.
\end{equation}
Furthermore, if $\tilde{\Omega}_{(n)}:= \frac{\sqrt{c_k}}{2a(t)^3\jmath(t)}+\frac{\gamma_n(t)}{k^{2n+1}}+\mathcal{O}(k^{-2n-3})$ then for $k\rightarrow \infty$ and all $n\in \mathbb{N}_0$
\begin{align}
\label{eqn_proof_hada_prop2a}
\frac{\dot{\tilde{\Omega}}_{(n)}}{\tilde{\Omega}_{(n)}} &= -3H-\frac{\dot{\jmath}}{\jmath}  + \frac{a (H\gamma_n +\dot{\gamma}_n)}{k^{2n+2}}+\mathcal{O}(k^{-2n-4})\\[0.1cm]
\label{eqn_proof_hada_prop2b}
\frac{\ddot{\tilde{\Omega}}_{(n)}}{\tilde{\Omega}_{(n)}} &= 2\left(\frac{k^2}{a^2} + m^2 -\frac{c_k}{4 a^6 j^2} + \frac{15}{2} H^2 + \frac{1}{4} R + \frac{9}{2} H  \frac{\dot{\jmath}}{\jmath} +\frac{3}{4} \frac{\dot{\jmath}^2}{\jmath^2}\right)\nonumber\\[-0.05cm]
&\hspace{0.5cm} +\,\frac{a}{k^{2n+2}}\,\left(\ddot{\gamma}_n-(3H^2+R/6)\gamma_n\right)+\mathcal{O}(k^{-2n-4}).
\end{align}
\end{lemma}
\begin{lemma}
\label{lem_proof_hada_prop2}
Let $a(t)$ be a solution to the consistency equation (\ref{eqn_scalefactor_cond}). Then $\Omega_{(n)}$ has the asymptotic behaviour
\begin{equation}
\label{eqn_proof_hada_prop1}
\forall n \in \mathbb{N}_0:\quad\Omega_{(n)}(t;k) = \frac{\sqrt{c_k}}{2 a(t)^3 \jmath(t)}+\frac{\gamma_n(t)}{k^{2n+1}}+\mathcal{O}(k^{-2n-3})\quad \text{for}\quad k \rightarrow \infty.
\end{equation} 
\end{lemma}
\begin{proof}
The proofs of these two lemmas are given in the appendix.
\end{proof}
\begin{theorem}
\label{lem_hadamard_prop_1}
Let $a(t)$ be a solution of the consistency equation (\ref{eqn_scalefactor_cond}), $q_k$ the associated seed solution (\ref{def_q(c>0)}) or (\ref{def_q(c<0)}) and $T_k^{(n)}$ an adiabatic vacuum state of order $n$ $-$ all defined on the closed time interval $I\subset \mathbb{R}$. Suppose furthermore that $q_k$ and $T^{(n)}_k$ are connected via the Bogoliubov transformation $q_k = \alpha_n(k)T^{(n)}_k+\beta_n(k)\,\bar{T}^{(n)}_k$. Then it holds, for large $k$,
\begin{equation*}
\abs{\alpha_n(k)}-1 = \mathcal{O}(k^{-2n-2}) \quad \text{and} \quad \abs{\beta_n(k)} = \mathcal{O}(k^{-2n-2}),
\end{equation*}
i.e. $q_k: I \rightarrow \mathbb{C}$ define a Hadamard state.
\end{theorem} 
\begin{proof}
By our assumption $q_k(t)$ and $T_k^{(n)}(t)$ are, for all $t\in I$, connected via a Bo\-gol\-iu\-bov transformation. In particular, this holds true for the initial time $t_1$, i.e.
\begin{align*}
q_k(t_1) \!	&=\! \alpha_n(k)T^{(n)}_k(t_1)\!+\!\beta_n(k) \bar{T}^{(n)}_k(t_1) = \alpha_n(k)W_{(n)}(t_1;k)\!+\!\beta_n(k)\bar{W}_{(n)}(t_1;k)\\[0.2cm]
\dot{q}_k(t_1)\!	&=\!  \alpha_n(k)\dot{T}^{(n)}_k(t_1)\! +\!\beta_n(k)\dot{\bar{T}}^{(n)}_k(t_1) =\alpha_n(k)\dot{W}_{(n)}(t_1;k)\!+\!\beta_n(k)\dot{\bar{W}}_{(n)}(t_1;k),
\end{align*}
where the dot denotes the derivative with respect to the time $t$. Without loss of generality we choose in the definition of $W_{(n)}$, cf. equation (\ref{def_adiabatic_state_Wn}), the so far unspecified parameter $t_0$ such that it equals $t_1$. Thereby the above initial conditions can be written in the simpler form 
\begin{align*}
q_k(t_1) 	&= \dfrac{\alpha_n(k)+\beta_n(k)}{\sqrt{2a(t_1)^3\Omega_{(n)}(t_1)}}\\[0cm]
\dot{q}_k(t_1)	&= -(\alpha_n(k)+\beta_n(k))\,\dfrac{3H(t_1)+\frac{\dot{\Omega}_{(n)}(t_1)}{\Omega_{(n)}(t_1)}}{2\sqrt{2a(t_1)^3\, \Omega_{(n)}(t_1)}}+\dfrac{\text{i}\,(\alpha_n(k)-\beta_n(k))\,\Omega_{(n)}(t_1)}{2\sqrt{2a(t_1)^3\, \Omega_{(n)}(t_1)}}.
\end{align*}
Replacing $q_k(t_1)$, $\dot{q}_k(t_1)$ by their explicit expressions resulting from (\ref{def_q(c>0)}) or (\ref{def_q(c<0)}) $-$ note that $\theta(t_1)=0$ $-$ and solving this system for $\alpha_n(k)$ and $\beta_n(k)$ we obtain
\begin{equation*}
\begin{array}{r c l}
\alpha_n(k) 	& = & \dfrac{1}{2}\,\sqrt{\dfrac{2 a^3 \jmath}{\sqrt{c_k}\,\Omega_{(n)}}}\left(\left(\dfrac{\sqrt{c_k}}{2 a^3 \jmath} + \Omega_{(n)}\right) - \text{i} \left(3 H + \dfrac{\dot{\jmath}}{\jmath} +\dfrac{\dot{\Omega}_{(n)}}{\Omega_{(n)}}\right)\right)\\[0.5cm]
\beta_n(k)	& = & \dfrac{1}{2}\,\sqrt{\dfrac{2 a^3 \jmath}{\sqrt{c_k}\,\Omega_{(n)}}}\left(\left(\Omega_{(n)}-\dfrac{\sqrt{c_k}}{2 a^3 \jmath} \right) + \text{i} \left(3 H + \dfrac{\dot{\jmath}}{\jmath} +\dfrac{\dot{\Omega}_{(n)}}{\Omega_{(n)}}\right)\right),
\end{array}
\end{equation*}
where all functions are evaluated at the initial time $t=t_1$. After we have taken the squared modulus of $\beta_n(k)$ we plug in the expansions (\ref{eqn_asymp_c/2aj}), (\ref{eqn_proof_hada_prop2a}) and (\ref{eqn_proof_hada_prop1}), which yields
\begin{align*}
\abs{\beta_n(k)}^2	& =  \dfrac{1}{4\left(\frac{k}{a}+\mathcal{O}(k^{-1})\right)^2}\left(\left(\dfrac{\gamma_n}{k^{2n+1}}+\mathcal{O}(k^{-2n-3})\right)^2+\left(\mathcal{O}(k^{-2n-2})\right)^2\right)\\
\\[-0.3cm]
			& =  \dfrac{a^2\,\gamma_n^2}{4\,k^{4n+4}}+\mathcal{O}(k^{-4n-5}).
\end{align*} 
If we finally take the square root we attain at the claimed order for $\abs{\beta_n}$. This also automatically induces the right order for $1-\abs{\alpha_n(k)}$ since $\abs{\alpha_n(k)}^2-\abs{\beta_n(k)}^2 = 1$.
\end{proof}
$\,$
\subsection{No-Go Theorem for the Chaplygin Gas Equation of State on Spacetimes with Non-Positive Scalar Curvature}
\label{sec_NoGoTheorem}
In the previous two subsections we have demonstrated how to obtain all HIH states once one seed (Hadamard) state is given. Furthermore, we have given a criterion on the scale factor sufficient to entail the existence of such a seed state. In this subsection we combine these results to make a statement about the existence of Hadamard states that obey the CG EoS. We will start with the spacetimes of Example $\emph{2}$ above and comment on the results with regard to their relation to the CG case.
\begin{theorem}
\label{thm_no-go-CG}
Let $(M,g)$ be a spatially flat FRW spacetime with scale factor $a\in C^\infty$ such that its scalar curvature is $R= \text{const.}\leq 0$ and let $\phi$ be the quantized minimally coupled scalar field of mass $m = \sqrt{-\frac{1}{6}R}$.\\
Then the only HIH state which induces a Chaplygin Gas equation of state $\state{p(t)}{\omega} = -\frac{A}{\state{\varrho(t)}{\omega}}$ with some $A>0$ is the Bunch-Davies state in the de Sitter case $a(t)=\beta_2\,\Exp{\frac{m}{\sqrt{2}}\,t}$, $\beta_2>0$, yielding $A=\frac{c^2\,m^8}{256\pi^4}$, $c\in \mathbb{R}$. For other choices of $a(t)$, no HIH state fulfills the Chaplygin Gas equation of state at all times.
\end{theorem}

\begin{proof}
The spacetime and the quantized scalar field under consideration are precisely the ones from Example $\emph{2}$ of the previous section, i.e. we have to distinguish five different scenarios, three for $m>0$ and two for $m=0$.\\
So, in the proof's first part we calculate the energy density and pressure for an HIH state, where we only use the assumption $R=-6m^2$. Afterwards, we check the theorem's claim for each case separately. We begin with the de Sitter case and show by analyzing the ``late'' time behavior of the energy density and pressure that a CG EoS cannot be realized except for the Bunch-Davies state. The remaining four cases are treated similarly. 

In order to compute the energy density and pressure we begin with the expression for the energy density per mode, cf. (\ref{eqn_energydensity/preassure_mode}). By theorem \ref{thm_characterisation_hihs} the most general mode function inducing an HIH state is given by
\begin{equation}
\label{eqn_proof_nogo_CGEoS_1}
T_k(t)= \sqrt{1+\alpha(k)^2}\,q_k(t)+\alpha(k)\,\Exp{-\text{i}\psi(k)}\,\bar{q}_k(t),
\end{equation}
where $\alpha(k)$ is polynomially bounded and ERD while $\psi$ is an integrable real-valued function. The mode functions $q_k$ on the other hand must be the ones associated to a Hadamard state. But, since by Example $\emph{2}$ the scale factor under consideration is a solution to the consistency equation (\ref{eqn_scalefactor_cond}), this is by theorem \ref{lem_hadamard_prop_1} indeed the case for the seed solution (\ref{def_q(c>0)}), which reads
\begin{equation}
\label{eqn_proof_nogo_CGEoS_2}
q_k(t) = \frac{1}{\sqrt{2k}\,a(t)}\,\Exp{\text{i}\,k\,\eta(t)} \qquad \text{with} \qquad \eta(t) := \integ{t_0}{t}{\frac{1}{a(\tau)}}{\tau}
\end{equation}
being the conformal time. Using these particular $q_k$ and $T_k$ we can easily compute $\varrho_k(t)$ and obtain
\begin{align*}
\varrho_k(t) = 	&\;\frac{1+2\,\Xi(k)}{4 k\, a(t)^4}\left(1 + 2 \alpha(k)^2\right) \left(2 k^2 + a(t)^2\left(m^2 + H(t)^2\right)\right)\\[0.1cm]
		&\;\;+\,\frac{1 + 2 \Xi(k)}{2 k\, a(t)^2}\,\alpha(k)\,\sqrt{1 + \alpha(k)^2}\left(m^2 + H(t)^2\right) \cos(2 k\,\eta(t)+\psi(k))\\[0.1cm]
		&\;\;+\,(1 + 2 \Xi(k))\,\frac{\alpha(k)}{a(t)^3}\,\sqrt{1 + \alpha(k)^2}\,H(t)\, \sin(2 k\, \eta(t)+\psi(k)).
\end{align*}
Now we need to calculate the Hadamard series $\mathfrak{h}^{\varrho}(t,k)$ to renormalize $\varrho_k$. Since $R=-6m^2$ and therefore the field is effectively conformally coupled and massless, the logarithmically diverging part of $\mathfrak{h}^{\varrho}(t,k)$ vanishes so that
\begin{align*}
\mathfrak{h}^\varrho(t,k) = 	& \;\frac{k}{2a(t)^4}+\frac{H(t)^2+m^2}{4a(t)^2\,k}-\frac{\pi\,\delta_0(\mathbf{k})}{120}\left(H(t)^4+12H(t)^2 m^2-11m^4\right).
\end{align*}
Finally, for these particular spacetimes we find for the renormalization ambiguities, cf. equation (\ref{eqn_renorm_ambiguity_rw_spacetime}),   
\begin{align*}
C_{00} = &\, (D_1 \!-\!18 D_3 \!+\! 6D_4) m^4 \!+\! 3 (D_2 \!+\! 12 D_3 \!-\! 8 D_4) m^2 H(t)^2 \!+\! 24 D_4 H(t)^4\\[0.1cm]
   \equiv & \;\frac{1}{960\,\pi^2}\left(\tilde{D}_1\, m^4 + \tilde{D}_2\, m^2\, H(t)^2 + \tilde{D}_3\,H(t)^4\right),
\end{align*}
where $\tilde{D}_1, \tilde{D}_2, \tilde{D}_3 \in \mathbb{R}$ are relabeled renormalization degrees of freedom. That is, by putting all these results together and defining for $n\in \mathbb{N}$ the, in part, time-dependent moments
\begin{equation}
\label{eqn_def_xicosi}
\begin{array}{ r l}
 \xi_n &:= \ainteg{0}{\infty}{k^n\left(\Xi(k)+2\Xi(k)\,\alpha(k)^2 +\alpha(k)^2\right)}{k}, \\
\\[-0.3cm]
\alpha_n &:= \ainteg{0}{\infty}{k^n\,(1+2\Xi(k))\,\alpha(k)\,\sqrt{1+\alpha(k)^2}}{k},\\
\\[-0.3cm]
\co_n(\eta) &:=  \ainteg{0}{\infty}{k^n\,(1+2\Xi(k))\,\alpha(k)\,\sqrt{1+\alpha(k)^2}\,\cos\left(2k\,\eta+\psi(k)\right)}{k},\\
\\[-0.3cm]
\si_n(\eta) & :=  \ainteg{0}{\infty}{k^n\,(1+2\Xi(k))\,\alpha(k)\,\sqrt{1+\alpha(k)^2}\,\sin\left(2k\,\eta+\psi(k)\right)}{k},
\end{array}
\end{equation}
with $\xi_n \geq 0$, $-\alpha_n\leq \co_n(\eta) \leq \alpha_n$ and $-\alpha_n\leq \si_n(\eta) \leq \alpha_n$, we obtain for the energy density and pressure, cf. equation (\ref{eqn_renorm_energy}) with $k_0=0$ as well as equation (\ref{eqn_renorm_pressure}), 
\begin{align}
\state{\varrho(t)}{\omega}\! &=\! \dfrac{1}{960\pi^2}\,\left((\tilde{D}_1-11)m^4+(\tilde{D}_2+12)m^2\,H(t)^2+(1+\tilde{D}_3)H(t)^4\right) \nonumber\\
			  &\quad +\,\dfrac{H(t)^2+m^2 }{4 \pi^2 a(t)^2}\,(\xi_1+\co_1(\eta))\,+ \dfrac{\xi_3}{2 \pi^2 a(t)^4}+ \dfrac{H(t)\,\si_2(\eta)}{2 \pi^2 a(t)^3}\nonumber\\[0.2cm]
\state{p(t)}{\omega}\! &=\! \frac{1}{2880\pi^2}\!\left(\!(9\!-\!3\tilde{D}_1\!-\!2\tilde{D}_2\!-\!\tilde{D}_3) m^4\!+\!(\tilde{D}_2\!+\!8)m^2H(t)^2\!+\!(5\!+\!\!\tilde{D}_3)H(t)^4\right)\nonumber\\[0.1cm]
\label{eqn_renorm_energydensity/pressure_r6m2}
			& \quad+ \dfrac{H(t)^2\!-\!m^2}{4\pi^2 a(t)^2}\!\left(\xi_1\!+\!\co_1(\eta)\right)\!+\! \dfrac{\xi_3\!-\!2\co_3(\eta)}{6 \pi^2 a(t)^4}\!+\! \dfrac{H(t)\,\si_2(\eta)}{2 \pi^2 a(t)^3}.  
\end{align}
\emph{First case: $a(t)= \beta_2\,\Exp{\frac{m}{\sqrt{2}}\,t}$}. In order to prove that there is no possibility of fulfilling the CG EoS $\state{\varrho}{\omega}\state{p}{\omega}=-A<0$ at all times, we will consider various (scaled) limits of this EoS for $t\rightarrow \infty$, i.e. we will make use of the necessary condition
\begin{equation*}
\forall n\in \mathbb{N}_0: \qquad \lim_{t\rightarrow \infty}\;a(t)^n\left(\state{\varrho(t)}{\omega}\,\state{p(t)}{\omega}+A\right) = 0 
\end{equation*}
and show that it is violated. Introducing the new constant $c:=\frac{1}{240}\, ( 4\tilde{D}_1 + 2\tilde{D}_2 + \tilde{D}_3-19) \in \mathbb{R}$, which is the remaining renormalization degree of freedom for this spacetime, we can take the product of $\state{\varrho}{\omega}$ and $\state{p}{\omega}$, order the result by inverse powers of the scale factor and obtain
\begin{align}
\label{eqn_dS_varrho*p}
\state{\varrho}{\omega}\,\state{p}{\omega} &=  -\frac{c^2\, m^8}{256 \pi^4}\! -\! \frac{c\, m^6(\xi_1\!+\!\co_1(\eta))}{32 \pi^4\,a(t)^2}\!-\!\frac{c m^4(\xi_3 \!+\! \co_3(\eta)\!+\! \frac{9}{4 c} (\xi_1\! +\! \co_1(\eta))^2)}{48 \pi^4\, a(t)^4}\nonumber\\[0.1cm]
				& \quad+\,\frac{m^3\,\si_2(\eta)\, (\xi_1 + \co_1(\eta))}{8 \sqrt{2} \pi^4\,a(t)^5} +\frac{m^2 (\si_2(\eta)^2 - (\xi_1 + \co_1(\eta))\,\co_3(\eta))}{8 \pi^4\, a(t)^6}\nonumber\\[0.1cm]
				& \quad+\,\frac{m\,\si_2(\eta)\,(2 \xi_3 - \co_3(\eta))}{6 \sqrt{2} \pi^4\, a(t)^7} + \frac{\xi_3\,(\xi_3 - 2 \co_3(\eta))}{12 \pi^4\, a(t)^8}.
\end{align} 
Before we can take the limits announced above we observe that  $\co_n(\eta)$ and $\si_n(\eta)$ are continuous functions, due to the dominant convergence theorem. Furthermore, we will denote $\lim\limits_{t\rightarrow \infty} \eta(t) \equiv \eta_\infty \in \mathbb{R}\cup\{\pm \infty\}$. As $\xi_n$, $\co_n(\eta)$ and $\si_n(\eta)$ are also always finite while $a(t)$ is diverging for $t\rightarrow \infty$ we find for the first limit
\begin{equation*}
0 = \lim_{t\rightarrow \infty}\left(\state{\varrho(t)}{\omega}\,\state{p(t)}{\omega}+A\right) =-\frac{c^2\, m^8}{256 \pi^4}+ A, 
\end{equation*}
which uniquely fixes the CG parameter $A$. To take the next limits scaled with $a(t)^n$, $n=2,3,4,$  we need to employ l'H\^{o}pital's rule, as for $n>2$ the $a(t)^n$ factor diverges while $(\state{\varrho(t)}{\omega}\,\state{p(t)}{\omega}+A)$ tends to zero for $t\rightarrow \infty$ $-$ note that, since the functions $\alpha(k)$ and $\Xi(k)$ are ERD we can differentiate in $\co_n$ and $\si_n$ under the integral. Hence, the various limits yield
\begin{align}
0& = \lim_{t\rightarrow \infty}a(t)^2\left(\state{\varrho(t)}{\omega}\,\state{p(t)}{\omega}+A\right) = - \frac{c\,m^6}{32 \pi^4}\,\,(\xi_1+\co_1(\eta_\infty))\nonumber\\
0& = \lim_{t\rightarrow \infty}a(t)^3\left(\state{\varrho(t)}{\omega}\,\state{p(t)}{\omega}+A\right) = -\frac{c\,m^5}{8 \sqrt{2} \pi^4}\,\si_2(\eta_\infty)\nonumber\\
\label{eqn_nogotheorem_proof_2}
0&= \lim_{t\rightarrow \infty}a(t)^4\left(\state{\varrho(t)}{\omega}\,\state{p(t)}{\omega}+A\right) = \; -\,\frac{c\, m^4}{48 \pi^4}\, (\xi_3 - 5 \co_3(\eta_\infty)).
\end{align}
Now we have all requirements to show that $a(t)^5\left(\state{\varrho(t)}{\omega}\,\state{p(t)}{\omega}+A\right)$ almost always diverges. By proceeding in the same way as for the previous limits we see that 
\begin{align*}
&\lim_{t\rightarrow \infty}a(t)^5 \left(\state{\varrho(t)}{\omega}\,\state{p(t)}{\omega}+A\right)\\
&\, =\!-\;\!\! \frac{c\;\! m^6}{32 \pi^4}\!\lim_{t\rightarrow \infty}\!a(t)^3(\xi_1\!+\!\co_1(\eta))\!-\!\frac{c\;\! m^4}{48 \pi^4}\!\lim_{t\rightarrow \infty}\! a(t)(\xi_3\! +\! \co_3(\eta)\!+\! \frac{9}{4 c} (\xi_1\! +\! \co_1(\eta))^2)\\[0.25cm]
&\, = \! \frac{c\, m^4}{48 \pi^4}\,\lim_{t\rightarrow \infty} a(t)\,(\co_3(\eta)-\xi_3).
\end{align*}
But due to equation (\ref{eqn_nogotheorem_proof_2}) the function $\co_3$ has to converge to $\frac{1}{5}\,\xi_3$ and thus the previous limit diverges unless $\xi_3 \equiv 0$. This is only possible if one chooses $\alpha=\Xi=0$, i.e. the (unique) Bunch-Davies state $\omega_\text{BD}$. In this state the energy density and pressure become constant and therefore (trivially) obey the CG EoS. So, a CG EoS is in general not possible on an FRW background with scale factor $a(t)= \beta_2\,\Exp{\frac{m}{\sqrt{2}}t}$.\\[0.3cm]
\emph{Second case: $a_\text{s}(t)=\beta_2\sqrt{\sinh(\sqrt{2}m\,t+\beta_1)}$ or $a_\text{c}(t)=\beta_2\sqrt{\cosh(\sqrt{2}m\,t+\beta_1)}$.}
To show that there is no CG EoS inducing HIH state for these two scale factors, one can repeat the same reasoning as above $-$ namely that there is a power $n\in \mathbb{N}$ such that $a_\text{c/s}(t)^n\left(\state{\varrho(t)}{\omega}\,\state{p(t)}{\omega}+A\right)$ diverges for all states. As these scale factors are more complex the computations become quite involved but nevertheless accomplish the aim. In particular, for both scale factors there is no state simultaneously yielding a constant energy density and pressure.\\[0.3cm]
\emph{Third case: $a(t) = \beta_2\,\sqrt{t+\beta_1}$}. Plugging in this scale factor into the expressions for $\state{\varrho}{\omega}$ and $\state{p}{\omega}$, cf. equation (\ref{eqn_renorm_energydensity/pressure_r6m2}), setting $m=0$ and using that $2H(t)= a(t)^{-2}$ we immediately get
\begin{align*}
\state{\varrho(t)}{\omega}& = \frac{1 + D_3}{15360 \pi^2\,a(t)^8} + \frac{\xi_3}{2\pi^2\, a(t)^4}+\frac{\si_2(\eta)}{4\pi^2\, a(t)^5} + \frac{\xi_1 + \co_1(\eta)}{16\pi^2\, a(t)^6}\;\overset{t\rightarrow \infty}{\longrightarrow}\; 0\\
\\[-0.3cm]
\state{p(t)}{\omega}& =  \frac{5 + D_3}{46080 \pi^2 a(t)^8} + \frac{\xi_3-2\co_3(\eta)}{6 \pi^2\, a(t)^4} + \frac{\si_2(\eta)}{4\pi^2\, a(t)^5}+ \frac{\xi_1 + \co_1(\eta)}{16 \pi^2\, a(t)^6}\; \overset{t\rightarrow \infty}{\longrightarrow}\; 0,
\end{align*} 
which obviously does not allow for a CG EoS.\\[0.3cm]
\emph{Fourth case: $a(t) = \beta_2$.} For the massless scalar field on Minkowski space the energy density and pressure are given by $\state{\varrho(t)}{\omega} = \xi_3/(2\pi^2\beta_2^4)$ and $\state{p(t)}{\omega} = (\xi_3-\co_2(\eta))/(6\pi^2\beta_2^4)$, respectively. As the energy density is constant in time so has to be the pressure in order to allow for a consistent equation of state. That is, $\co_2(\eta) \equiv 0$. But then $\state{\varrho(t)}{\omega}$ and $\state{p(t)}{\omega}$ obey the well-known equation of state $\state{p}{\omega}=\frac{1}{3}\state{\varrho}{\omega}\geq0$. 
\end{proof}
\begin{remark}
If the scale factor possesses a zero at $t_0 \in \mathbb{R}\cup \{\pm\infty\}$ we can also consider the scaled limits $\lim\limits_{t\rightarrow t_0}\;a(t)^{n}\left(\state{\varrho(t)}{\omega}\,\state{p(t)}{\omega}+A\right) = 0,\;n\in\mathbb{N},$ to construct a contradiction. For instance in the de Sitter case, we obtain by equation (\ref{eqn_dS_varrho*p}) for the limit\footnote{Note that by the Riemann-Lebesgue lemma $\lim\limits_{\eta(t)\rightarrow \pm\infty} \co_n(\eta(t)) = 0.$} $0= \lim\limits_{t\rightarrow -\infty}\;a(t)^{8}\left(\state{\varrho(t)}{\omega}\,\state{p(t)}{\omega}+A\right) = \xi_3^2/(12\pi^4)$ and hence the same necessary condition $\xi_3=0$ as for the limit $t \rightarrow \infty.$
\end{remark}
The theorem basically shows that on spacetimes with a non-positive constant scalar curvature the late time scaling behavior of energy density and pressure are incompatible with the scaling behavior required by the CG EoS. Here the existence of one exception, namely the Bunch-Davies state, is merely due to the high symmetry of de Sitter space. In all other cases this symmetry is not present and hence a state yielding a constant energy density and pressure does not exist. In particular, such a state would rather mimic the properties of a cosmological constant than the properties of the CG EoS.\\
Now the CG scale factor (\ref{eqn_cg_scalefactor}) converges at late times to a de Sitter scale factor $-$ in terms of the conformal time $\eta(t):= \int^\infty_t a_\text{cg}(\tau)^{-1}\text{d}\tau$ one can for instance show that 
\begin{equation}
\label{eqn_latetime_acg}
\lim\limits_{\eta\rightarrow 0}\;\!\! \frac{1}{\eta^4}\!\left(\!(a_\text{cg}\circ t)(\eta)\!-\!\frac{\sqrt{3}}{\sqrt{8\pi} A^\frac{1}{4}\eta}\right)\;\!\!\!=\;\!\!0  \text{ and } \lim\limits_{\eta\rightarrow 0}\!\;\! \frac{1}{\eta^5}\!\left(\frac{1}{6}(R\circ t)(\eta)\!+\! \frac{16}{3}\pi\sqrt{A}\right)\;\!\!\! =\!\;\!0
\end{equation}
with $a_\text{dS}(\eta)=\sqrt{\frac{3}{8\pi\sqrt{A}}}\,\eta^{-1}$ being the de Sitter scale factor with respect to conformal time. Therefore, one might hope that one can approximate the late time behavior of the energy density and pressure by the respective late time behavior on a de Sitter space and thus show that there is at most one candidate state, namely the one that converges to the Bunch-Davies state as this is the only state on de Sitter space with the right scaling behavior. Using the latter equation of (\ref{eqn_latetime_acg}), one can even argue, similar to Dapiaggi, Moretti \& Pinamonti in the proof of their theorem 4.2 in \cite{dap_mor_pin},\footnote{One recasts the temporal KG equation by the ansatz $T_k = \varphi_k/a_\text{cg}$ and a change to conformal time into $\varphi_k''(\eta)+k^2\varphi_k(\eta) + (a\circ t)(\eta)^2(m^2+\frac{1}{6}(R\circ t)(\eta))\varphi_k(\eta) = 0$. The last term can be considered as a perturbation as it vanishes sufficiently fast for $\eta \rightarrow 0$ and thus allows the construction of a solution in terms of the Dyson series.} that there is a state $T_k(t):=\varphi_k(\eta(t))/a_\text{cg}(t)$ on the CG spacetime such that 
\begin{equation*}
\lim\limits_{\eta \rightarrow 0}\,\frac{\text{d}^n}{\text{d}\eta^n}\left(\varphi_k(\eta)-\chi_k(\eta)\right) = 0 \qquad \quad (n=0,...,5),
\end{equation*}        
where $\chi_k(\eta)=\frac{1}{\sqrt{2k}} \Exp{-\text{i}k \eta}$ is the Bunch-Davies state on the associated ``limiting'' de Sitter space written with respect to conformal time. This means there actually is a state on the CG spacetime that converges for $t\rightarrow \infty$ to the Bunch-Davies. However, in order to make this heuristic argument rigorous one would at least need to show that $T_k(t)$ is also a Hadamard state $-$ which is not obvious and could not be achieved so far.\footnote{The techniques used by Dapiaggi, Moretti and Pinamonti, cf. \cite{dap_mor_pin_2}, in order to prove that their bulk-to-boundary state is actually Hadamard if the background spacetime converges for $t\rightarrow 0$ sufficiently fast towards a de Sitter spacetime heavily rely on the fact that those spacetimes possess a (geodesically complete) past cosmological horizon. A similar horizon does however not exist if the spacetimes converge for $t \rightarrow \infty$ towards a de Sitter space, as it is considered here.}\\
Another problem arises for the case of the CG spacetime if one approaches the Big Bang singularity at $t=0$. In contrast to the $R=-6m^2$-scale factors the functions $\mathfrak{h}^{\varrho}_{-3}(t)$ and $\mathfrak{h}^{p}_{-3}(t)$ of the Hadamard expansion, cf. (\ref{eqn_h_rho}) in combination with (\ref{eqn_renorm_pressure}), respectively, do not vanish for $a_\text{cg}$ so that one needs to use the complete formula (\ref{eqn_renorm_energy}), again in combination with (\ref{eqn_renorm_pressure}), to calculate the energy density and pressure. In particular, the term proportional to $\log a_\text{cg}(t)$ probably remains in the final expressions for $\state{\varrho}{\omega}$ and $\state{p}{\omega}$ $-$ note that this term cannot be canceled by some special choice of the renormalization ambiguities but, if at all possible, would have to be canceled by a state dependent counterpart. That means, for $t\rightarrow 0$ \emph{both} $\state{\varrho}{\omega}$ and $\state{p}{\omega}$ presumably diverge, which is incompatible with an EoS of the form $\state{p}{\omega}=-A\,\state{\varrho}{\omega}^{-1}$. Regarding this problem it would be interesting to also have a statement similar to the one in theorem \ref{thm_no-go-CG} for the scale factor $a(t) = \sinh^\frac{2}{3}(\frac{3m}{2\sqrt{2}}\,t)$, since this scale factor shows the same asymptotic behavior for early and late times as the CG scale factor $a_\text{cg}$. By Example $\emph{3}$ of section \ref{sec_Construction_of_a_Hadamard_State_on_certain_spacetimes} we would even have a seed state to repeat the analysis. Unfortunately, the calculations become too extensive so that no conclusion could yet be drawn.

\section{Conclusion and Discussion}

In this paper we tried to answer the question whether it is possible that quantum matter can induce a Chaplygin gas equation of state (CG EoS). If this was indeed the case one would not need to introduce some mysterious dark energy of unknown nature in order to explain the observed accelerated expansion of the universe.\\
We approached the problem within the framework of quantum field theory on curved spacetimes. Then the task is to see if there are physical states, i.e. Hadamard states, which solve the semi-classical Einstein equation such that the CG EoS holds. As we worked on a spatially flat FRW background, which is spatially homogeneous and isotropic, we assumed the same symmetry for the states under consideration $-$ note that this is sufficient to obtain an energy-momentum tensor of perfect fluid type. However, not every Hadamard state is homogeneous and isotropic and vice versa not every homogeneous and isotropic state is Hadamard. So in the first step we deduced the transformation behavior between two (spatially) homogeneous and isotropic Hadamard (HIH) states. It turned out that the coefficients of the connecting Bogoliubov transformation have to be of essentially rapid decay. Then we presented a new approach for the explicit construction of an HIH state. While this method is restricted to a small class of scale factors $-$ yet some very important scale factors are contained in this class $-$ the obtained Hadamard state is of a simple form.\\
Unfortunately, it is not possible to explicitly solve the temporal Klein-Gordon equation on the Chaplygin gas (CG) spacetime, i.e. an FRW spacetime equipped with the CG scale factor $a_\text{cg}$. This solution is however needed in order to explicitly construct the HIH states and thus to decide whether their associated energy density and pressure obey the CG EoS on this particular spacetime. Therefore, we analyzed a special cases, namely the scalar field on $R=\text{const.}\leq 0$ spacetimes. In contrast to the CG spacetime these spacetimes allow for an explicit construction of the energy density and pressure, as long as one chooses the scalar field's mass to be $m=\sqrt{-\frac{1}{6}R}$, and hence are available for a direct analysis $-$ as in this case the scalar field is effectively conformally coupled, the same analysis should apply to the massless conformally coupled scalar field up to some minor alterations, for instance one would have to use a different EMT. By investigating the energy density's and pressure's late time behavior we were able to show that except for the de Sitter space, which belongs to the class of spacetimes considered here, there is no HIH state inducing the CG EoS at all times. In the de Sitter case only the Bunch-Davies state yields a CG EoS. But, the energy density and pressure have to be constants. In particular note that the existence of this exception is only possible due to the high symmetry of the de Sitter space, a symmetry which is in general not present on other FRW spacetimes, either.\\
Finally, we presented some (heuristic) arguments that seem to contradict the existence of HIH states entailing the CG EoS on the CG spacetime.\\
However, as our argumentation was entirely based on the scaling behavior of the energy density and pressure at late and early times it is, strictly speaking, still possible that there are HIH states, which allow for an approximate CG EoS in the intermediate time regime. While this might still be a worthwhile possibility from a conceptual point of view we see less use for those states in their application to cosmology. The major advantage of the CG EoS was that it provides a single and simple ``mechanism'' to explain the occurrence of a Big Bang, a matter dominated phase and a ``dark energy era'' in the evolution of the universe in a \emph{unified} way. This therefore requires that this EoS holds at \emph{all} times. Even if one wants to use the CG EoS just as an explanation for dark energy it has to be valid at least at late times.


\begin{appendix}
\section{Appendix}
\subsubsection*{Proof of lemma \ref{lem_h_rho/p}}
\begin{proof}
The leading and subleading orders $\mathfrak{h}^{\varrho}_{2n-1}(t)$, $n=-1,0,1$, have been already determined in \cite{eltz_gott}, thus we can confine ourselves to the zeroth order contribution $\mathfrak{h}^\varrho_0(t)$. This is obtained by a very tedious but straight-forward calculation. For this reason we only describe the general line of reasoning $-$ a more detailed explanation of this computation is however given in \cite{zschoche}.\\
Let us start with the expansion of the Hadamard parametrix as defined in (\ref{eqn_hadamard_prop_2}) with $\varepsilon=0$ and its derivatives. Due to the spatial homogeneity and isotropy of a flat FRW spacetime the squared geodesic distance $\sigma(q,q')$ depends on $q$ and $q'$ solely via $t, t'$ and $z:=\norm{\mathbf{x-y}}$, the same holds for the Hadamard recursion coefficients $u(q,q')$ and $v_j(q,q')$ . Furthermore, it was shown in \cite{moretti_0} that $u(q,q')$ and $v_j(q,q')$ are actually symmetric so that $\mathcal{H}^\text{s}_{l,0}=\mathcal{H}_{l,0}$ and since the energy-momentum operator only contains derivatives up to 2nd order, it suffices to consider the case $l=1$ $-$ all terms stemming from contributions with $l>1$ vanish after taking the coincidence limit. Therefore, we have 
\begin{equation*}
\mathcal{H}^\text{s}_{1,0}(t,t',z) =  -\frac{u}{4\pi^2\,\sigma}-\frac{1}{4\pi^2}\,\left(v_0+v_1\,\sigma\right)\, \log\left(-\frac{\sigma}{\lambda^2}\right).
\end{equation*} 
Hence, to determine $\mathfrak{h}^\varrho_0(t)$ one needs to calculate the zeroth order contribution from an expansion of $\left[H_1(t,t',z)\right]_{t'=t}$, $\left[\nab{0}\!\otimes\nab{0'}H_1(t,t',z)\right]_{t'=t}$ and $\sum^3_{i=1}\left[\nab{i}\!\otimes\nab{i'}H_1(t,t',z)\right]_{t'=t}$ with respect to $z$ at $z_0=0$ and add them up according to the definition of $\mathfrak{h}^\varrho$. This can be however achieved in a straight-forward manner employing the expansions for $\sigma(t,t',z),\, u(t,t',z),\, v_0(t,t',z)$ and $v_1(t,t',z)$ as given in \cite{eltz_gott}. 
\end{proof}

\subsubsection*{Proof of lemma \ref{lem_HIH_pressure}}
\begin{proof}
To shorten notations we sketch the proof for a massless scalar field in a pure HIH state. First we have to apply $T^\text{can}_{i i'}$ to the two-point function $\mathscr{W}^{\omega}_2(q,q')$ of the HIH state, cf. equation (\ref{eqn_2pt_his_3}), which yields in the coincidence limit $q' \rightarrow q$ the ``unrenormalized pressure'':
\begin{align}
\clim{T^\text{can}_{i i'}\mathscr{W}^{\omega}_2} & = \lim_{\delta \rightarrow 0}\,\int_{\mathbb{R}^3}\frac{a(t)^2\, \Exp{-\delta k}}{2(2\pi)^3}\left(\abs{\dot{T}_k(t)}^2\!+\!\left(\frac{k^2_i\!-\!\sum_{j\neq i}k_j^2}{a(t)^2}\right) \abs{T_k(t)}^2\right) \text{d}^3\mathbf{k}\nonumber\\
	& =: a(t)^2\,\lim\limits_{\delta \rightarrow 0} p^\delta_i(t)\nonumber 
\end{align} 
Note that, since the integrand of $p_i^\delta(t)$ was not renormalized, the limit $\delta \rightarrow 0$ is not well-defined. However, $p_i^\delta(t)$ is finite and as $T_k(t)$ depends on $\mathbf{k}$ only via $k=\abs{\mathbf{k}},$ which entails, for all $i,j = 1,2,3$ and $\delta >0,$ the identity
\begin{equation*}
\integ{\mathbb{R}^3}{}{k^2_i\, \abs{T_k(t)}^2\,\Exp{-\delta k}}{^3\mathbf{k}} = \integ{\mathbb{R}^3}{}{k^2_j\, \abs{T_k(t)}^2\,\Exp{-\delta k}}{^3\mathbf{k}},
\end{equation*} 
we have 
\begin{equation*}
p_1^\delta(t) = \integ{\mathbb{R}^3}{}{\,\frac{1}{2(2\pi)^3}\left(\abs{\dot{T}_k(t)}^2-\frac{k^{\,2}_1}{a(t)^2}\, \abs{T_k(t)}^2\right) \Exp{-\delta k}}{^3\mathbf{k}} = p_2^\delta(t)  = p_3^\delta(t).
\end{equation*}
That is, the ``unrenormalized pressure'' is isotropic. As the spacetime and thus the Hadamard parametrix as well as the renormalization ambiguities $C_{a b}$ are also isotropic (and spatially homogeneous), this property remains conserved even after renormalization, i.e. it will hold $\state{p_1(t)}{\omega}=\state{p_2(t)}{\omega}=\state{p_3(t)}{\omega}$ $-$ the generalization to $m>0$ and $\Xi(k)>0$ is straight-forward.\\[0.2cm]
Now that we established the EMT's isotropy, equation (\ref{eqn_renorm_pressure}) follows immediately from the covariant conservation of the EMT, which holds by construction.
\end{proof}

\subsubsection*{Proof of lemma \ref{lem_prop_erd_1}}
\begin{proof}
Claim \emph{$(1)$} is trivial. The forward direction of the second assertion follows directly from the observation that, for $k=(k_1,...,k_n)\in \mathbb{R}^n$, we have $k_i \leq \abs{k}$. Thus, for any multi-index $(l_1,...,l_n)\in \mathbb{N}^n$ with $\sum_i l_i =: N$, we can conclude $\int_\Omega k_1^{l_1} \cdot ... \cdot k_n^{l_n} f(k)\text{d}\mu \leq \int_\Omega \abs{k}^N \abs{f(k)}\text{d}\mu < \infty.$ The inverse statement can be obtained if one expands the power $\abs{k}^{2m}$, takes into account that $\abs{f}=f$ and thus expresses $\norm{\abs{k}^{2m} f(k)}_1$ as a finite sum of finite moments of $f$. This in turn immediately yields the finiteness for $\norm{\abs{\,\cdot\,}^{2m-1}\,f}_1$, too, since
\begin{align*}
\norm{\abs{\,\cdot\,}^{2m-1}f}_1& = \integ{B_1(0)}{}{k^{2m-1}  \abs{f(\mathbf{k})}}{^n\mathbf{k}}+\!\!\!\!\integ{\mathbb{R}^n\backslash B_1(0)}{}{\!\!k^{2m-1} \abs{f(\mathbf{k})}}{^n\mathbf{k}}\\
						    & \leq \norm{f}_1+\norm{\abs{\,\cdot\,}^{2m} f}_1 < \infty.
\end{align*}
The third claim is again trivial. While the first part of the fourth assertion follows directly from the triangle inequality $\abs{f+g}\leq \abs{f}+\abs{g}$ the second one is a consequence of $\abs{f}\leq \abs{f+g}$ since $g$ is positive. To show claim \emph{$(5)$} we estimate $ \norm{(1+\abs{\,\cdot\,})^m\, g}_1 = \norm{f^{-1}\,(1+\abs{\,\cdot\,})^m\, f\cdot g}_1 \leq \norm{(1+\abs{\,\cdot\,})^m\, f\cdot g}_1\; \norm{f^{-1}}_\infty \leq \norm{(1+\abs{\,\cdot\,})^m\, f\cdot g}_1\; \inf\limits_{k\in \Omega} \abs{f(k)} <\infty$, where the first inequality is due to the H\"older inequality with $p=1$ and $q=\infty$. Using claim \emph{$(1)$} this proves \emph{$(5)$}.\\
The last assertion is again basically a consequence of H\"older's inequality (HI). Namely, suppose that the assumptions of claim \emph{$(6)$} are fulfilled for some fixed $p\in(1,\infty)$ and that $l\in \mathbb{N}$ $-$ note that by claim \emph{(1)} the case $p=1$ is already proven. Then we have 
\begin{align*}
\norm{(1\!+\!\abs{\,\cdot\,})^m f}_1 &\! =\! \norm{(1\!+\!\abs{\,\cdot\,})^{m+l} f (1\!+\!\abs{\,\cdot\,})^{-l}}_1\!\overset{\text{HI}}{\leq}\! \norm{(1\!+\!\abs{\,\cdot\,})^{m+l} f}_p \norm{(1\!+\!\abs{\,\cdot\,})^{-l}}_{\frac{p}{p-1}} \\[0.2cm]
			     &\!\;\! =\!\;\! \norm{(1+\abs{\,\cdot\,})^{m+l}\,f}_p\,\norm{(1+\abs{\,\cdot\,})^{-\frac{l\,p}{p-1}}}_{1}^{\frac{p-1}{p}}.
\end{align*}
Since $-\frac{p}{p-1}<-1$, for $1<p<\infty$, we can always find a sufficiently large $l\in \mathbb{N}$ such that $\norm{(1+\abs{\,\cdot\,})^{-\frac{l\,p}{p-1}}}_{1} < \infty$. On the other hand $\norm{(1+\abs{\,\cdot\,})^{m+l}\,f}_p$ is finite by assumption. That is $\norm{(1+\abs{\,\cdot\,})^m\,f}_1< \infty$ and by claim \emph{(1)} the function $f$ therefore has to be ERD.
\end{proof}

\subsubsection*{Proof of lemma \ref{lem_smoothness_moments}}
\begin{proof}
Let $t > 0$ and $\mathbf{e} \in \mathbb{R}^n$ an arbitrary direction of $\mathbb{R}^n$. Then we define for $h:\mathbb{R}^n \rightarrow \mathbb{C}$ the $m$-th balanced difference by $\Delta^{t\cdot \mathbf{e}}_1\, h(\mathbf{y})  := h(\mathbf{y}+t\cdot \mathbf{e}) - h(\mathbf{y}-t\cdot \mathbf{e}) \quad \text{and}\quad \Delta^{t\cdot \mathbf{e}}_{m+1}\, h(\mathbf{y})  := \Delta^{t\cdot \mathbf{e}}_1\,\Delta^{t\cdot \mathbf{e}}_m\, h(\mathbf{y}).$ By induction it can be shown that this recursive definition is equivalent to the explicit form
\begin{equation*}
\Delta^{t\cdot \mathbf{e}}_m\, h(\mathbf{y}) = \sum^m_{l=0} (-1)^l\, \binom{m}{l} \, h\left(\mathbf{y}+(m-2l)\,t\cdot \mathbf{e}\right).
\end{equation*}
If we apply this result to the function $h(\mathbf{y})=\Exp{\text{i}\,\mathbf{k}\mathbf{ y}},\;\mathbf{k}\in\mathbb{R}^n,$ this yields
\begin{align}
\label{eqn_delta_sinus}
\Delta^{t\cdot \mathbf{e}}_m\, h(\mathbf{y}) & = \Exp{\text{i}\;\! \mathbf{k y}} \left(\Exp{\text{i} t\, \mathbf{k e}}-\Exp{-\text{i} t \,\mathbf{k e}}\right)^m = \Exp{\text{i} \mathbf{k y}} \left(2\text{i}\sin(t\,\mathbf{k e})\right)^m.
\end{align}  
Now suppose $\omega \in C^\infty(\mathbb{R}^n)$. Hence all derivatives of $\omega$ exist, in particular the $(2s)$-th derivative in the direction of $\mathbf{e}$ at $\mathbf{y}=0$, and we have
\begin{equation*}
\infty > M \equiv \abs{\left.\nabla^{2s}_{\mathbf{e}}\omega\right|_{\mathbf{y}=0}} = \underset{t\rightarrow 0}{\lim} \abs{\frac{\Delta^{t \cdot \mathbf{e}}_{2s}\,\omega(0)}{(2t)^{2s}}}.
\end{equation*}
Using the definition of $\omega$ and equation (\ref{eqn_delta_sinus}), the balanced difference at $\mathbf{y}=0$ is given by 
\begin{equation*}
\Delta^{t \cdot \mathbf{e}}_{2s}\,\omega(0)=(-1)^s \lim\limits_{\varepsilon \rightarrow 0} \integ{\mathbb{R}^n}{}{(2\sin(t\,\mathbf{k e}))^{2s}\varphi(\mathbf{k})\,\text{e}^{-\varepsilon\,k}}{^n \mathbf{k}}.
\end{equation*}
Because $\varphi$ is non-negative, we can then estimate $M$ from below with
\begin{align}
M & = \underset{t\rightarrow 0}{\lim}\,\lim\limits_{\varepsilon\rightarrow 0} \integ{\mathbb{R}^n}{}{\left(\frac{\sin(t\,\mathbf{k e})}{t}\right)^{2s} \varphi(\mathbf{k})\,\text{e}^{-\varepsilon\,k}}{^n \mathbf{k}}\nonumber\\
    & \geq \underset{t\rightarrow 0}{\lim} \,\lim\limits_{\varepsilon\rightarrow 0} \integ{B_R(0)}{}{\left(\frac{\sin(t\,\mathbf{k e})}{t}\right)^{2s} \varphi(\mathbf{k})\, \text{e}^{-\varepsilon\,k}}{^n \mathbf{k}}\nonumber\\
    & =\underset{t\rightarrow 0}{\lim} \integ{B_R(0)}{}{\left(\frac{\sin(t\,\mathbf{k e})}{t}\right)^{2s} \varphi(\mathbf{k})}{^n\mathbf{ k}} = \integ{B_R(0)}{}{(\mathbf{k e})^{2s} \varphi(\mathbf{k})}{^n \mathbf{k}} 
\label{eqn_smoothness_moments_proof_1}
\end{align}
for \emph{all} positive $R$, where $B_R(0)$ is the $n$-dimensional ball around the origin with radius $R$. Choosing $\mathbf{e}$ as one of the axis' directions $\mathbf{e}_i$ we immediately get the existence of the $(2s)$-th $k_i$-moment of $\varphi$, since the last integral in (\ref{eqn_smoothness_moments_proof_1}) can be viewed as a strictly increasing real-valued function of the ball's radius $R$ that is bounded from above by $M$ and therefore has to converge for $R\rightarrow \infty$. The same argument as in the proof of lemma \ref{lem_prop_erd_1}.$\emph{(2)}$ then yields the existence of all $k_i$-moments with an order smaller than $2s$. By iterating this procedure for different directions we can create any moment of desired order and see that it is always possible to bound these by the modulus of the corresponding derivative at $\mathbf{y}=0$. Since $\varphi$ was non-negative we finally apply lemma \ref{lem_prop_erd_1}.\emph{(2)} and thus prove the first statement.\\
If we on the other hand assume that $\varphi$ is ERD we can pull the regularizing $\varepsilon$-limit into the integral by the dominant convergence theorem, i.e. we have $\omega(\mathbf{z}) = \mathcal{F}_0[\varphi](\mathbf{z})$. For the same reason we can also differentiate underneath the integral wherein each derivative generates a monomial in $k_i$. And since 
\begin{equation*}
\abs{\frac{\partial^{\abs{N}}}{\partial z^N} \integ{\mathbb{R}^n}{}{\varphi(\mathbf{k})\,\Exp{\text{i}\,\mathbf{k z}}}{^n \mathbf{k}}} \leq \integ{\mathbb{R}^n}{}{k^{\abs{N}}\,\abs{\varphi(\mathbf{k})}}{^n \mathbf{k}}<\infty,   
\end{equation*}
where $N\in \mathbb{N}^n$ is a multi-index, every derivative is finite. 
\end{proof}

\subsubsection*{Proof of lemma \ref{lem_proof_hada_prop1}}

\begin{proof}
Let us start with the proof of equation (\ref{eqn_asymp_c/2aj}). Since $a(t)$ solves the consistency equation, $\frac{c_k}{4a(t)^6\jmath(t)^2}$ is given by
\begin{equation}
\label{eqn_defc/2aj}
\frac{c_k}{4 a^6 \jmath^2} =  \frac{k^2}{a^2} + m^2 - \frac{\dot{\jmath}^2}{4 \jmath^2} + \frac{1}{2} \left(3 H\,\frac{\dot{\jmath}}{\jmath} + \frac{\ddot{\jmath}}{\jmath}\right)
\end{equation}
and that is, we need the asymptotic behaviour of $\dot{\jmath}/\jmath$ and $\ddot{\jmath}/\jmath$. By definition of $\jmath$, see equation (\ref{def_j}), we have 
\begin{align}
\jmath &=\! \frac{6H \mathfrak{s}+\dot{\mathfrak{s}}}{H (m^2+2\omega^2)}\nonumber\\[0.15cm]
       &\!\!\;=\! \frac{12 k^4\! +\! k^2 a^2 (12 m^2\! - \!R)\! -\!\!\!\; \frac{a^4}{4} (18 m^4\!\!\!\; +\!\!\;\! 2 m^2\!\!\;R\! -\!\frac{R^2}{6}\! -\!2(6 m^2\! +\! R)H^2\! -\! 5 H\! \dot{R} \!-\!\ddot{R})}{12 k^3 a^2 (2 k^2 + 3 m^2 a^2)}\nonumber\\[0.15cm]
\label{eqn_PQ1}
     &\equiv\! \frac{A_4(t)\,k^4+A_2(t)\,k^2+A_0(t)}{k^3\,(B_2(t)\,k^2+B_0(t))}\equiv \frac{P(t,k)}{k^3\,Q(t,k)}.
\end{align} 
and therefore 
\begin{equation}
\label{eqn_PQ2}
\frac{\dot{\jmath}}{\jmath} = \frac{\dot{P}Q-P\dot{Q}}{Q P} \quad \text{and} \quad \frac{\ddot{\jmath}}{\jmath} = \frac{2P\dot{Q}^2+\ddot{P}Q^2-2\dot{P}\dot{Q}Q-PQ\ddot{Q}}{Q^2 P} 
\end{equation}
are rational functions with the numerator and the denominator being of order $k^6$. Employing, for large $k$, the expansion scheme
\begin{equation*}
\frac{a_1 k^6 + b_1 k^4 + c_1 k^2 + d_1}{a_2 k^6 + b_2 k^4 + c_2 k^2 + d_2} = \frac{a_1}{a_2} + \frac{a_2 b_1 - a_1 b_2}{a_2^2 k^2} +\mathcal{O}(k^{-4})
\end{equation*}
we find that we need to determine the two highest orders in $k$ in the above expressions for $\dot{\jmath}/\jmath$ and $\ddot{\jmath}/\jmath$. These can however be read off from (\ref{eqn_PQ1}) and (\ref{eqn_PQ2}) and we obtain
\begin{align}
\frac{\dot{\jmath}}{\jmath} &= \frac{\dot{A}_4 B_2 - \dot{B}_2 A_4}{A_4 B_2}+\frac{B_2^2 (A_4 \dot{A}_2-A_2 \dot{A}_4)- A_4^2(B_2 \dot{B}_0 - B_0 \dot{B})}{A_4^2 B_2^2\,k^2} +\mathcal{O}(k^{-4}) \nonumber\\[0.1cm]
\label{eqn_expansion_j'/j}
 & =  -2 H - \frac{a^2}{12k^2}\,(2 H (6 m^2 + R) +\dot{R})+\mathcal{O}(k^{-4})
\end{align}
as well as, without stating the general expansion scheme, 
\begin{equation}
\label{eqn_j'/j_asymp}
\frac{\ddot{\jmath}}{\jmath} = 8 H^2 + \frac{R}{3} + \frac{a^2}{36 k^2}\,(6 m^2 R + R^2 + 24 H^2 (6 m^2 + R) -3\ddot{R}) + \mathcal{O}(k^{-4}).
\end{equation}
Plugging these results into (\ref{eqn_defc/2aj}) we get
\begin{equation}
\label{eqn_proof_c/2a3j}
\frac{c_k}{4a^6\jmath^2} = \frac{k^2}{a^2} + (m^2 + \frac{1}{6}R) + \mathcal{O}(k^{-2}) \;\; \Longrightarrow\;\;
\frac{\sqrt{c_k}}{2a^3\jmath} = \frac{k}{a} +  \frac{a}{12 k}\,(6 m^2 + R)+\mathcal{O}(k^{-3}).
\end{equation}
Now suppose $\tilde{\Omega}_{(n)} = \sqrt{c_k}/(2a^3\jmath)+ \gamma_n/k^{2n+1}+\mathcal{O}(k^{-2n-3})$. If we take the time derivative of $\tilde{\Omega}_{(n)}$, divide it afterwards by $\tilde{\Omega}_{(n)}$ and use that $2a^3\jmath/\sqrt{c_k} = a/k +\mathcal{O}(k^{-3})$, which follows from (\ref{eqn_proof_c/2a3j}), we get
\begin{eqnarray*}
\dfrac{\dot{\tilde{\Omega}}_{(n)}}{\tilde{\Omega}_{(n)}}& = & \frac{-3H - \dfrac{\dot{\jmath}}{\jmath}+\dfrac{2a^3\jmath}{\sqrt{c_k}} \left(\dfrac{\dot{\gamma}_n}{k^{2n+1}}+\mathcal{O}(k^{-2n-3}) \right)}{1+\dfrac{2a^3\jmath}{\sqrt{c_k}}\left(\dfrac{\gamma_n}{k^{2n+1}}+\mathcal{O}(k^{-2n-3})\right)}\\[0.2cm]
	&  = &\frac{ -3H - \dfrac{\dot{\jmath}}{\jmath}+\dfrac{a\,\dot{\gamma}_n}{k^{2n+2}}+\mathcal{O}(k^{-2n-4})}{1+ \dfrac{a \gamma_n}{k^{2n+2}}+\mathcal{O}(k^{-2n-4})}\\[0.2cm]
	& = &-3H-\frac{\dot{\jmath}}{\jmath}+\frac{a\gamma_n}{k^{2n+2}}\left(3H +\frac{\dot{\jmath}}{\jmath}+\frac{\dot{\gamma}_n}{\gamma_n}\right) + \mathcal{O}(k^{-2n-4})\\[0.2cm]
& \overset{(\ref{eqn_expansion_j'/j})}{=} &-3H-\frac{\dot{\jmath}}{\jmath}+\frac{a}{k^{2n+2}}\left(H\gamma_n +\dot{\gamma}_n\right) + \mathcal{O}(k^{-2n-4}).
\end{eqnarray*}
The remaining proof of (\ref{eqn_proof_hada_prop2b}) works analogously: calculate the second derivative of $\tilde{\Omega}_{(n)}$ and replace the occurring second time derivative of $\jmath$ by (\ref{eqn_scalefactor_cond}), then pull the overall factor $\sqrt{c_k}/a^3\jmath$ out of the numerator and denominator of $\ddot{\tilde{\Omega}}_{(n)}/\tilde{\Omega}_{(n)}$ and expand the obtained expression for large $k$, which yields
\begin{align*}
\frac{\ddot{\tilde{\Omega}}_{(n)}}{\tilde{\Omega}_{(n)}} \overset{(\ref{eqn_asymp_c/2aj})}{=}&\, \frac{2\left(\frac{k^2}{a^2}\! +\! m^2\! -\!\frac{c_k}{4 a^6 j^2}\! +\! \frac{15}{2} H^2\! +\! \frac{1}{4} R\! +\! \frac{9}{2} H  \frac{\dot{\jmath}}{\jmath} \!+\!\frac{3}{4} \frac{\dot{\jmath}^2}{\jmath^2}\right) \!+\!\frac{a\ddot{\gamma}_n}{k^{2 n+2}}\!+\!\mathcal{O}(k^{-2n-4})}{1 + \frac{a \gamma_n}{k^{2 n+2}}+\mathcal{O}(k^{-2n-4})}\\[0.25cm]
&\hspace{-0.6cm}= \,2\,\left(\frac{k^2}{a^2} + m^2 -\frac{c_k}{4 a^6 j^2} + \frac{15}{2} H^2 + \frac{1}{4} R + \frac{9}{2} H  \frac{\dot{\jmath}}{\jmath} +\frac{3}{4} \frac{\dot{\jmath}^2}{\jmath^2}\right)\\
 &\hspace{-0.2cm}\! +\!\!\left.\left.\frac{a}{k^{2n+2}}\right(\!\ddot{\gamma}_n\!-\!2\gamma_n \right(\underset{\overset{(\ref{eqn_proof_c/2a3j})}{=}- \frac{1}{6} R+\mathcal{O}(k^{-2}) }{\underbrace{\frac{k^2}{a^2}\! +\! m^2\! -\!\frac{c_k}{4 a^6 j^2}}} \!+\!\!\! \underset{\overset{(\ref{eqn_j'/j_asymp})}{=} -6H^2+\mathcal{O}(k^{-2}) }{\underbrace{\frac{9H}{2} \frac{\dot{\jmath}}{\jmath}\! +\! \frac{3}{4} \frac{\dot{\jmath}^2}{\jmath^2}}}\left.\left.\hspace{-0.45cm}+\,\frac{15}{2} H^2\! +\! \frac{R}{4} \right)\!\!\right)\\[0.25cm]
&\hspace{-0.6cm}=\, 2\,\left(\frac{k^2}{a^2} + m^2 -\frac{c_k}{4 a^6 j^2} + \frac{15}{2} H^2 + \frac{R}{4} + \frac{9H}{2}\frac{\dot{\jmath}}{\jmath} +\frac{3}{4} \frac{\dot{\jmath}^2}{\jmath^2}\right) \\
& \hspace{5cm}+ \frac{a}{k^{2n+2}}\,(\ddot{\gamma}_n-(3H^2+R/6)\gamma_n).
\end{align*}
\end{proof}

\subsubsection*{Proof of lemma \ref{lem_proof_hada_prop2}}

\emph{Proof.} We will prove the lemma by induction. For $n=0$, $\Omega_{(0)}$ equals $\omega$ and, hence, by equation (\ref{eqn_asymp_c/2aj}) of the previous lemma, we can conclude that
\begin{align*}
\Omega_{(0)} -\frac{\sqrt{c_k}}{2a^3\jmath}& = \sqrt{\frac{k^2}{a^2}+m^2} - \frac{k}{a}-\frac{a}{12k}\,(6m^2+R)+\mathcal{O}(k^{-3})\\
	     & = \frac{k}{a}\!+\!\frac{a\,m^2}{2k}\!- \!\frac{k}{a}\!-\!\frac{a}{12k}\,(6m^2+R)\!+\!\mathcal{O}(k^{-3}) = -\frac{a}{12k}\,R+\mathcal{O}(k^{-3}),
\end{align*}
which proves the base case. So, let us assume we have shown for some arbitrary $n \in \mathbb{N}_0$ that $\Omega_{(n)} = \sqrt{c_k}/(2a^3\jmath)+\gamma_n/k^{2n+1}+\mathcal{O}(k^{-2n-3})$. Using the definition of $\Omega_{(n+1)}^2$ and applying lemma \ref{lem_proof_hada_prop1}, which is applicable due to the induction hypothesis, we obtain
\begin{align}
\Omega_{(n+1)}^2 & = \omega^2+\frac{3}{4}\,H^2+\frac{1}{4}\,R+\frac{3}{4}\,\left(\frac{\dot{\Omega}_{(n)}}{\Omega_{(n)}}\right)^2-\frac{1}{2}\,\frac{\ddot{\Omega}_{(n)}}{\Omega_{(n)}}\nonumber\\[0.cm]
		    & = \frac{c_k}{(2a^3\jmath)^2} -\frac{a}{2\,k^{2n+2}}\left(\ddot{\gamma}_n+3 H \dot{\gamma}_n+\nicefrac{R}{6}\,\gamma_n\right) +\mathcal{O}(k^{-2n-4}).\nonumber \\[-0.2cm]
\nonumber
\end{align}
In order to finish the induction we finally employ the expansion of $\sqrt{1\! -\! \frac{b}{x^{2n+2}}\!+\!\mathcal{O}(k^{-2n-4})}= 1\!-\!\frac{b}{2}\frac{1}{x^{2n+2}}\!+\! \mathcal{O}(k^{-2n-4})$ and $2a^3\jmath/\sqrt{c_k} = \frac{k}{a} +\mathcal{O}(k^{-1})$ to the last equation and thus have achieved the desired form of
\begin{align*}
\Omega_{(n+1)} &=\, \frac{\sqrt{c_k}}{2a^3 \jmath}- \frac{\sqrt{c_k}}{2a^3\jmath}\left(\frac{a}{4\,k^{2n+2}}\left(\ddot{\gamma}_n+3 H \dot{\gamma}_n+\nicefrac{R}{6}\,\gamma_n\right) + \mathcal{O}(k^{-2n-4})\right)\\[0.2cm]
			   &=\, \frac{\sqrt{c_k}}{2a^3 \jmath}- \frac{a^2}{4\,k^{2n+3}}\left(\ddot{\gamma}_n+3 H \dot{\gamma}_n+\nicefrac{R}{6}\,\gamma_n\right) + \mathcal{O}(k^{-2n-5}).
\end{align*}
\hfill $\square$

\end{appendix}

\subsection*{Acknowledgment}
The author gratefully acknowledges financial support by the International Max Planck Research School (IMPRS). He would like to thank B. Eltzner as well as N. Pinamonti for very helpful discussions and the two referees for their very constructive comments on the manuscript. He also thanks his supervisor Rainer Verch for valuable  discussions and his highly appreciated help in writing this article, which was carried out in partial fulfillment of the author's PhD project.


\end{document}